\documentclass[pra,twocolumn, preprintnumbers,amsmath,amssymb]{revtex4}

\let\oldaddcontentsline\addcontentsline
\newcommand{\stoptocentries}{\renewcommand{\addcontentsline}[3]{}}
\newcommand{\starttocentries}{\let\addcontentsline\oldaddcontentsline}

\usepackage{graphicx, color, graphpap}
\usepackage{amsmath}
\usepackage{enumitem}
\usepackage{amssymb}
\usepackage{amsthm}
\usepackage{pstricks}
\usepackage{float}
\usepackage{multirow}
\usepackage{hyperref}
\usepackage[T1]{fontenc}

\long\def\ca#1\cb{} 

\newcommand{\abs}[2][]{#1| #2 #1|}
\newcommand{\braket}[2]{\langle #1 \hspace{1pt} | \hspace{1pt} #2 \rangle}
\newcommand{\braketq}[1]{\braket{#1}{#1}}
\newcommand{\ketbra}[2]{| \hspace{1pt} #1 \rangle \langle #2 \hspace{1pt} |}

\newcommand{\ketbraq}[1]{\ketbra{#1}{#1}}
\newcommand{\bramatket}[3]{\langle #1 \hspace{1pt} | #2 | \hspace{1pt} #3 \rangle}
\newcommand{\bramatketq}[2]{\bramatket{#1}{#2}{#1}}
\newcommand{\norm}[2][]{#1| \! #1| #2 #1| \! #1|}
\newcommand{\nbox}[2][9]{\hspace{#1pt} \mbox{#2} \hspace{#1pt}}

\newcommand{\ket}[1]{|#1\rangle}               
\newcommand{\bra}[1]{\langle #1|}              
\newcommand{\dya}[1]{\ket{#1}\!\bra{#1}}
\newcommand{\dyad}[2]{\ket{#1}\!\bra{#2}}        
\newcommand{\ip}[2]{\langle #1|#2\rangle}      

\newcommand{\AC}{\mathcal{A}}
\newcommand{\BC}{\mathcal{B}}
\newcommand{\CC}{\mathcal{C}}
\newcommand{\DC}{\mathcal{D}}
\newcommand{\EC}{\mathcal{E}}
\newcommand{\FC}{\mathcal{F}}
\newcommand{\GC}{\mathcal{G}}

\newcommand{\IC}{\mathcal{I}}

\newcommand{\PC}{\mathcal{P}}
\newcommand{\QC}{\mathcal{Q}}
\newcommand{\RC}{\mathcal{R}}
\newcommand{\SC}{\mathcal{S}}

\newcommand{\VC}{\mathcal{V}}

\newcommand{\ZC}{\mathcal{Z}}

\newcommand{\Tr}{{\rm Tr}}

\newcommand{\ave}[1]{\langle #1\rangle}               
\renewcommand{\geq}{\geqslant}
\renewcommand{\leq}{\leqslant}
\newcommand{\mte}[2]{\langle#1|#2|#1\rangle }

\newcommand{\ot}{\otimes}
\newcommand{\ad}{^\dagger}

\newcommand*{\id}{\openone}

\newcommand*{\g}{\text{guess}}

\newcommand{\rhob}{\overline{\rho}}

\newcommand{\rhot}{\tilde{\rho}}
\newcommand{\bs}{\textsf{BS}}
\newcommand{\qbs}{\textsf{QBS}}
\newcommand{\pbs}{\textsf{PBS}}
\newcommand{\enh}{\text{enh}}

\newcommand{\al}{\alpha }

\newcommand{\kp}{\kappa }

\newcommand{\Lm}{\Lambda }

\newcommand{\sg}{\sigma }

\newcommand{\En}{F}
\newcommand{\Az}{M}
\newcommand{\Aw}{N}
\newcommand{\AAz}{M}
\newcommand{\AAw}{N}
\newcommand{\LLL}{L}
\newcommand{\HHH}{H}
\newcommand{\Ct}{\tilde{C}}

\newtheoremstyle{example}{\topsep}{\topsep}%
{}
{}
{\bfseries}
{.}
{   }
{\thmname{#1}\thmnumber{ #2}}
\theoremstyle{example}

\newtheorem{theorem}{Theorem}
\newtheorem{lemma}[theorem]{Lemma}

\theoremstyle{definition}

\newcommand{\footnoteremember}[2]{\footnote{#2}\newcounter{#1}\setcounter{#1}{\value{footnote}}}
\newcommand{\footnoterecall}[1]{\footnotemark[\value{#1}]}

\begin{document}

\title{Equivalence of wave-particle duality to entropic uncertainty}

\author{Patrick J. Coles}
\affiliation{Centre for Quantum Technologies, National University of Singapore, 2 Science Drive 3, 117543 Singapore}

\author{J\k{e}drzej Kaniewski}
\affiliation{Centre for Quantum Technologies, National University of Singapore, 2 Science Drive 3, 117543 Singapore}

\author{Stephanie Wehner}
\affiliation{Centre for Quantum Technologies, National University of Singapore, 2 Science Drive 3, 117543 Singapore}

\begin{abstract}
Interferometers capture a basic mystery of quantum mechanics: a single particle can exhibit wave behavior, yet that wave behavior disappears when one tries to determine the particle's path inside the interferometer. This idea has been formulated quantitively as an inequality, e.g., by Englert and Jaeger, Shimony, and Vaidman, which upper bounds the sum of the interference visibility and the path distinguishability. Such wave-particle duality relations (WPDRs) are often thought to be conceptually inequivalent to Heisenberg's uncertainty principle, although this has been debated. Here we show that WPDRs correspond precisely to a modern formulation of the uncertainty principle in terms of entropies, namely the min- and max-entropies. This observation unifies two fundamental concepts in quantum mechanics. Furthermore, it leads to a robust framework for deriving novel WPDRs by applying entropic uncertainty relations to interferometric models. As an illustration, we derive a novel relation that captures the coherence in a quantum beam splitter.
\end{abstract}


\maketitle

\stoptocentries

\section*{INTRODUCTION}

When Feynman discussed the two-path interferometer in his famous lectures \cite{Feynman70}, he noted that quantum systems (quantons) display the behavior of both waves and particles and that there is a sort of competition between seeing the wave behavior versus the particle behavior. That is, when the observer tries harder to figure out which path of the interferometer the quanton takes, the wave-like interference becomes less visible. This tradeoff is commonly called wave-particle duality (WPD). Feynman further noted that this is ``a phenomenon which is impossible ... to explain in any classical way, and which has in it the heart of quantum mechanics. In reality, it contains the only mystery [of quantum mechanics].''

Many quantitative statements of this idea, so-called wave-particle duality relations (WPDRs), have been formulated \cite{EnglertPRL1996, JaegerEtAlPRA1995, PhysRevD.19.473, Greenberger1988391, EnglertIJQI2008, PhysRevA.79.052108, PhysRevA.87.022107, Qureshi01042013, PhysRevA.85.054101, Englert2000337, Banaszek:2013fk, JiaEtAlCPB2014}. Such relations typically consider the Mach-Zehnder interferometer for single photons, see Fig.~\ref{fgrMZ}. For example, a well-known formulation proven independently by Englert \cite{EnglertPRL1996} and Jaeger et al.~\cite{JaegerEtAlPRA1995} quantifies the wave behavior by fringe visibility $\VC$, and particle behavior by the distinguishability of the photon's path, $\DC$. (See below for precise definitions; the idea is that ``waves'' have a definite phase, while ``particles'' have a definite location, hence $\VC$ and $\DC$ respectively quantify how definite the phase and location are inside the interferometer.) They found the tradeoff:
\begin{equation}
\label{eqn1}
\DC^2 + \VC ^2 \leq 1
\end{equation}
which implies $\VC = 0$ when $\DC = 1$ (full particle behavior means no wave behavior) and vice-versa, and also treats the intermediate case of partial distinguishability.

It has been debated, particularly around the mid-1990's \cite{Englert:1995ly, Storey:1994zr, Wiseman:1995ys}, whether the WPD principle, closely related to Bohr's complementarity principle \cite{Bohr1928}, is equivalent to another fundamental quantum idea with no classical analog: Heisenberg's uncertainty principle \cite{Heisenberg}. The latter states that there are certain pairs of observables, such as position and momentum or two orthogonal components of spin angular momentum, that cannot simultaneously be known or jointly measured. Likewise there are many quantitative statements of this idea, known as uncertainty relations (URs)~(see, e.g., \cite{kennard1927quantum, Robertson, deutsch, BiaMyc75, MaassenUffink, EURreview1, RenesBoileau, BertaEtAl, ColesEtAlPRA2011, ColesColbeckYuZwolak2012PRL,TomRen2010}), and modern formulations typically use entropy instead of standard deviation as the uncertainty measure, so-called entropic uncertainty relations (EURs) \cite{EURreview1}. This is because the standard deviation formulation suffers from trivial bounds when applied to finite-dimensional systems \cite{deutsch}, whereas the entropic formulation not only fixes this weakness but also implies the standard deviation relation \cite{BiaMyc75} and has relevance to information-processing tasks.

At present the debate regarding wave-particle duality and uncertainty remains unresolved, to our knowledge. Yet Feynman's quote seems to suggest a belief that quantum mechanics has but one mystery and not two separate ones. In this article we confirm this belief by showing a quantitative connection between URs and WPDRs, demonstrating that URs and WPDRs capture the same underlying physics; see also \cite{DurrRempe2000, Busch20061} for some partial progress along these lines. This may come as a surprise, since Englert \cite{EnglertPRL1996} originally argued that \eqref{eqn1} ``does not make use of Heisenberg's uncertainty relation in any form''.  To be fair, the uncertainty relation that we show is equivalent to \eqref{eqn1} was not known at the time of Englert's paper, and was only recently discovered \cite{RenesBoileau, BertaEtAl, ColesEtAlPRA2011, ColesColbeckYuZwolak2012PRL,TomRen2010}. Specifically, we will consider EURs, where the particular entropies that are relevant to \eqref{eqn1} are the so-called min- and max-entropies used in cryptography \cite{KonRenSch09}.

In what follows we provide a general framework for deriving and discussing WPDRs - a framework that is ultimately based on the entropic uncertainty principle. We illustrate our framework by showing that several different WPDRs from the literature are in fact particular examples of EURs. Making this connection not only unifies two fundamental concepts in quantum mechanics, but also implies that novel WPDRs can be derived simply by applying already-proven EURs. Indeed we use our framework to derive a novel WPDR for an exotic scenario involving a ``quantum beam splitter'' \cite{PhysRevLett.107.230406, Kaiser02112012, Peruzzo02112012, TangEtAlPhysRevA.88.014103}, where testing our WPDR would allow the experimenter to verify the beam splitter's quantum coherence (see~\eqref{eqn23943521}).

\begin{figure}[t]
\begin{center}
\includegraphics[scale=0.95]{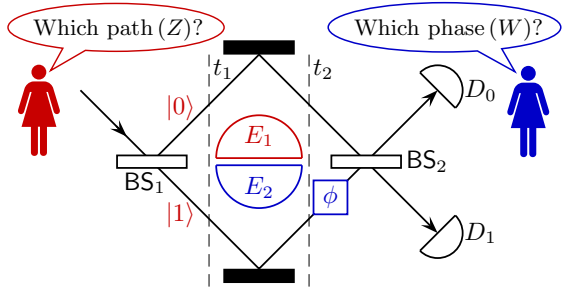}
\caption{%
Mach-Zehnder interferometer for single photons. Passing through the first beam splitter creates a superposition of which-path states, $\ket{0}$ and $\ket{1}$, at time $t_1$, then the system interacts with an environment $E=E_1 E_2$. Finally at time $t_2$ a phase shift $\phi$ is applied to the lower arm and the two beams are recombined on a second beam splitter. (While this is the typical setup, our framework also allows $E$ to play a more general role, e.g., being correlated to the photon before it enters the MZI.) Our complementary guessing game proceeds as follows. In one game (colored red) Alice tries to guess which of the two paths the photon took given that she has access to a portion of $E$ denoted $E_1$, which could be, e.g., a gas of atoms whose internal states record information about the presence of a photon. In the other game (colored blue), one of two phases, $\phi=\phi_0$ or $\phi=\phi_0+\pi$, is randomly applied to the lower interferometer arm and Alice tries to guess $\phi$ given that she has access to a different portion of $E$ denoted $E_2$, which could be, e.g., the photon's polarisation. We argue that WPDRs impose fundamental trade-offs on Alice's ability to win these two games.
\label{fgrMZ}}
\end{center}
\end{figure}

We emphasize that the framework provided by EURs is highly robust, and entropies have well-characterized statistical meanings. Note that current approaches to deriving WPDRs often involve brute force calculation of the quantities one aims to bound; there is no general, elegant method currently in use. Our approach simply involves judicial application of the relevant uncertainty relation. What's more, we emphasize that uncertainty relations can be applied to interferometers in two different ways. One involves preparation uncertainty, which says that a quantum state cannot be prepared having low uncertainty for two complementary observables, and it turns out this is the principle relevant to the original presentation of \eqref{eqn1} in \cite{EnglertPRL1996}. The other involves measurement uncertainty, which says that two complementary observables cannot be jointly measured \cite{Busch20061, PhysRevA.79.052108}, and we discuss why this principle is actually what was tested in some recent interferometry experiments \cite{PhysRevLett.100.220402, Kaiser02112012}.

\section*{RESULTS}

\subsection*{Framework}

\textit{Guessing games.---}We argue that a natural and powerful way to think of wave-particle duality is in terms of guessing games, and one's ability to win such games is quantified by entropic quantities. Specifically we consider \textit{complementary} guessing games, where Alice is asked to guess one of two complementary observables - a modern paradigm for discussing the uncertainty principle. In the Mach-Zehnder interferometer (MZI), see Fig.~\ref{fgrMZ}, this corresponds to either guessing which path the photon took, or which phase was applied inside the interferometer. The which-path and which-phase observables are complementary and hence the uncertainty principle gives a fundamental restriction stating that Alice cannot be able to guess both observables.

\textit{Binary interferometers.---}Our framework treats this complementary guessing game for binary interferometers. By \textit{binary}, we mean any interferometer where there are only two interfering paths, i.e., all other paths are classically distinguishable (from each other and from the two interfering paths). In addition to the MZI, this includes as special cases, e.g., the Franson interferometer \cite{PhysRevLett.62.2205} (see Fig.~\ref{fgrFran}) and the double slit interferometer (see Fig.~\ref{fgrDS}).  Note that binary interferometers go beyond interferometers with two physical paths. For example, in the Franson interferometer there are four possible paths but post-selecting on coincidence counts discards two of these paths, which are irrelevant to the interference anyway.

\begin{figure}[t]
\begin{center}
\includegraphics[scale=0.9]{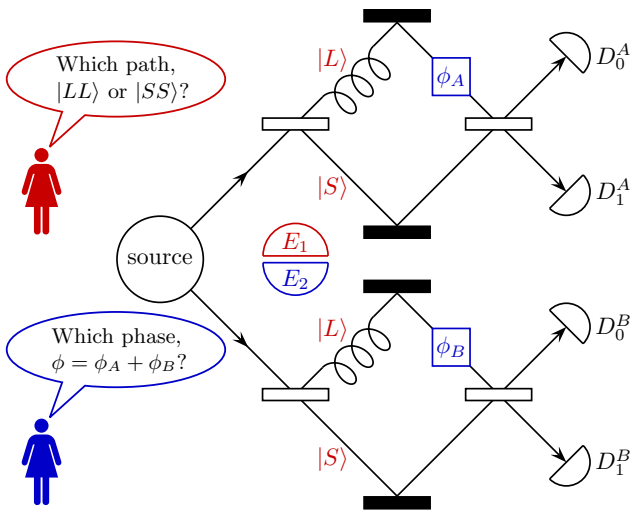}
\caption{%
In the Franson interferometer, the ``quanton'' consists of two photons. That is, a source produces time-energy entangled photons that each head separately towards a MZI that contains a long arm (depicted with extra loops) and a short arm. A simple model considers the four-dimensional Hilbert space associated with the four possible paths: $\ket{SS}$, $\ket{SL}$, $\ket{LS}$, and $\ket{LL}$ with $S$ = short path, $L$ = long path. Two of these dimensions are post-selected away by considering only coincidence counts, i.e., the photons arriving at the same time is inconsistent with the $\ket{SL}$ and $\ket{LS}$ paths. The remaining paths, $\ket{0}:=\ket{SS}$ and $\ket{1}:=\ket{LL}$, are indistinguishable in the special case of perfect visibility and they produce interference fringes as one varies $\phi := \phi_A +\phi_B$. Namely, the intensity of coincidence counts at detector pair ($D_0^A, D_0^B$) oscillates with $\phi$. Interaction with an environment system, or making the beam splitters asymmetric, may allow one to partially distinguish between $\ket{SS}$ and $\ket{LL}$, and our entropic uncertainty framework can be applied to derive a tradeoff, e.g., of the form of \eqref{eqn1}. This tradeoff captures the idea that Alice can either guess which path ($\ket{SS}$ vs.\ $\ket{LL}$) or which phase ($\phi =\phi_0 $ vs.\ $\phi =\phi_0+\pi$), but she cannot do both (even if she extracts information from other systems $E_1$ and $E_2$).
\label{fgrFran}}
\end{center}
\end{figure}

\textit{Particle observable.---}Now we link wave and particle behavior to knowledge of complementary observables. In the case of particle behavior, the intuition is that particles have a well-defined spatial location, hence ``particleness'' should be connected to knowledge of the path inside interferometer. For binary interferometers, there may be more than two physical paths but only two of these are interfering. Hence we only consider the two-dimensional subspace associated with the two which-path states of interest, denoted $\ket{0}$ and $\ket{1}$. This subspace can be thought of as an effective qubit, denoted $Q$, and the standard basis of this qubit:
\begin{equation}
\label{eqnZ11}
\text{which-path:}\hspace{8pt}Z = \{\ket{0},\ket{1}\}
\end{equation}
corresponds precisely to the which-path observable. For example, in the double slit (Fig.~\ref{fgrDS}), $\ket{0}$ and $\ket{1}$ are the pure states that one would obtain at the slit exit from blocking the bottom and top slits respectively.

\textit{Wave observable.---}Wave behavior is traditionally associated with having a large amplitude of intensity oscillations at the interferometer output. Indeed this has been quantified by the so-called \textit{fringe visibility}, see \eqref{eqnfringevis1}, but to apply the uncertainty principle we need to relate wave behavior to an observable inside the interferometer. Classical waves (e.g., water waves) are often modelled as having a \textit{well-defined phase} and being spatially \textit{delocalized}. The analog in our context corresponds to the quanton being in a equally-weighted superposition of which-path states. Hence eigenstates of the ``wave observable'' should live in the $XY$ plane of the Bloch sphere, so we consider observables on qubit $Q$ (the interfering subspace) of the form 
\begin{equation}
\label{eqnW11}
\text{which-phase:}\hspace{8pt}W=\{\ket{w_{\pm}}\},\quad\text{}\ket{w_{\pm}}=\frac{1}{\sqrt{2}} (\ket{0}\pm e^{i \phi_0 }\ket{1}).
\end{equation}
In terms of the guessing game, guessing the value of the wave (or which-phase) observable corresponds to guessing whether a phase of $\phi=\phi_0$ or $\phi = \phi_0+ \pi$ was applied inside the interferometer (see, e.g., Fig~\ref{fgrMZ}). While $\phi_0$ is a generic phase, its precise value will be singled out by the particular experimental setup. When the experimenter measures fringe visibility this corresponds to varying $\phi_0$ to find the largest intensity contrast, and mathematically we model this by minimizing the uncertainty within the $XY$ plane, see \eqref{eq:lackofwave}.

\subsection*{Entropic View}

Our entropic view associates a kind of behavior with the availability of a kind of information, or lack of behavior with missing information, as follows:
\begin{subequations}
\label{eq:partwave}
\begin{align}
\label{eq:lackofpart}&\text{lack of particle behavior:} \hspace{2pt}H_{\min}(Z |E_1)\\
\label{eq:lackofwave}&\text{lack of wave behavior:} \hspace{8pt} \min_{W\in XY}H_{\max}(W |E_2) 
\end{align}
\end{subequations}
where $H_{\min}$ and $H_{\max}$ are the min- and max-entropies, defined below in \eqref{eq:minmaxabsdef}, which are commonly used in quantum information theory, $Z$ is the which-path observable in \eqref{eqnZ11}, $W$ is the which-phase observable in \eqref{eqnW11} (whose uncertainty we optimize over the $XY$ plane of the Bloch sphere), and $E_1$ and $E_2$ are some other quantum systems that contain information and measuring these systems may help to reveal the behavior (e.g., $E_1$ could be a which-path detector and $E_2$ could be the quanton's internal degree of freedom). Note that we use the same symbols ($Z$, $W$, etc.) for the observables as for the random variables they give rise to. Full behavior (no behavior) of some kind corresponds to the associated entropy in \eqref{eq:partwave} being zero (one).  We formulate our general WPDR as
\begin{equation}
\label{eqnMainResult}
H_{\min}(Z |E_1)+ \min_{W\in XY}H_{\max}(W |E_2) \geq 1.
\end{equation}
This states that, for a binary interferometer, the sum of the ignorances about the particle and wave behaviors is lower bounded by 1 (i.e., 1 bit). Eq.~\eqref{eqnMainResult} constrains Alice's ability to win the complementary guessing game described above. If measuring $E_1$ allows her to guess the quanton's path, i.e., the min-entropy in \eqref{eq:lackofpart} is small, then even if she measures $E_2$ she still will not be able to guess the quanton's phase, i.e., the max-entropy in \eqref{eq:lackofwave} will be large (and vice-versa).

\begin{figure}[t]
\begin{center}
\includegraphics[scale=1]{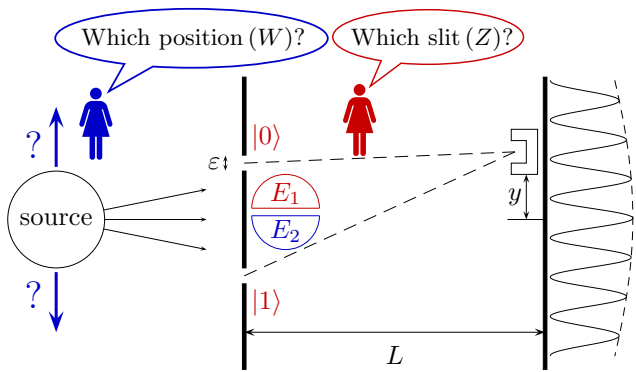}
\caption{%
For the double slit interferometer, the first game (colored red) involves Alice guessing which slit the quanton goes through, given that she can measure a system $E_1$ that has interacted with, and hence may contain information about, the quanton. In the second game (colored blue), Bob randomly chooses the source's vertical coordinate, and Alice tries to guess where the source was located, given some other system $E_2$ and given where the quanton was finally detected. Note that the source's location determines the relative phase between the which-slit states, $\ket{0}$ and $\ket{1}$, and we assume Bob chooses one of two possible locations such that the relative phase is either $0$ or $\pi$. Here, the state $\ket{0}$ ($\ket{1}$) is defined as the pure state at the slit exit that one would obtain from blocking the bottom (top) slit. Our framework provides a WPDR that constrains Alice's ability to win these complementary games. Furthermore, one can reinterpret the probability to win the second game, for the case where $E_2$ is trivial, in terms of the traditional fringe visibility. The latter quantifies the amplitude of intensity oscillations as one varies the detector location $y$. Note that varying $y$ changes the relative path lengths from the slits to the detector and hence is analogous to applying a relative phase $\phi$ between the two paths. [This assumes that the envelope function (dashed curve) associated with the interference pattern is flat over the range of $y$ values considered, which is often the case when $L$ is very large.] So the double slit fringe visibility is equivalent to the notion captured by $\VC$ in \eqref{eqnfringevis1}, where the detector is spatially fixed but a phase is varied, and we relate $\VC$ to our entropic measure of wave behavior in the Methods section.
\label{fgrDS}}
\end{center}
\end{figure}

To be clear, \eqref{eqnMainResult} is explicitly an entropic uncertainty relation, and it has been exploited to prove the security of quantum cryptography \cite{TLGR}. The usefulness of \eqref{eqnMainResult} for cryptography is due to the clear operational meanings of the min- and max-entropies \cite{KonRenSch09}, which naturally express the monogamy of correlations as they give the distances to being uncorrelated ($H_{\max}$) and being perfectly correlated ($H_{\min}$). One can replace these entropies with the von Neumann entropy in \eqref{eqnMainResult} and the relation still holds; however, the min- and max-entropies give more refined statements about information processing since they are also applicable to finite numbers of experiments. From~\cite{KonRenSch09}, the precise definitions of these entropies, for a generic classical-quantum state $\rho_{XB}$, are
\begin{subequations}
\label{eq:minmaxabsdef}
\begin{align}
\label{eq:minabsdef}H_{\min}(X|B)& = -\log p_{\g}(X|B),\\
\label{eq:maxabsdef}H_{\max}(X|B)& = \log p_{\text{secr}}(X|B),
\end{align}
\end{subequations}
where all logarithms are base 2 in this article. Here, $p_{\g}(X|B)$ denotes the probability for the experimenter to guess $X$ correctly with the optimal strategy, i.e., with the optimally helpful measurement on system $B$. Also, $p_{\text{secr}}(X|B) = \max_{\sg_B} F(\rho_{XB}, \id \ot \sg_B)^2$ quantifies the secrecy of $X$ from $B$, as measured by the fidelity $F$ of $\rho_{XB}$ to a state that is completely uncorrelated.

The fact that \eqref{eqnMainResult} can be thought of as a WPDR, and furthermore that it encompasses the majority of WPDRs found in the literature for binary interferometers, is our main result.

\section*{DISCUSSION}

To illustrate this, we consider the celebrated MZI, shown in Fig.~\ref{fgrMZ}, since most literature WPDRs have been formulated for this interferometer. In the simplest case one sends in a single photon towards a 50/50 (i.e., symmetric) beam splitter, $\bs_1$, which results in the state $\ket{+}=(\ket{0}+\ket{1})/\sqrt{2}$, then a phase $\phi$ is applied to the lower arm giving the state $(\ket{0}+e^{i\phi }\ket{1})/\sqrt{2}$. Finally the two paths are recombined on a second 50/50 beam splitter $\bs_2$ and the output modes are detected by detectors $D_0$ and $D_1$. Fringe visibility is then defined as
\begin{align}
\label{eqnfringevis1}
\text{fringe visibility:}\hspace{4pt}\VC:= \frac{p^{D_0}_{\max}-p^{D_0}_{\min}}{p^{D_0}_{\max}+p^{D_0}_{\min}},
\end{align}
where $p^{D_0}$ is the probability for the photon to be detected at $D_0$, $p^{D_0}_{\max}:= \max_{\phi} p^{D_0}$ maximizes this probability over $\phi$, whereas $p^{D_0}_{\min}:= \min_{\phi} p^{D_0}$.  In this trivial example one has $\VC = 1$. However many more complicated situations, for which the analysis is more interesting, have been considered in the extensive literature; we now illustrate how these situations fall under the umbrella of our framework with some examples.

\textit{$\PC$-$ \VC$ relation.---}As a warm-up, we begin with the simplest known WPDR, the predictability-visibility tradeoff. Predictability $\PC $ quantifies the \textit{prior} knowledge, given the experimental setup, about which path the photon will take inside the interferometer. More precisely, $\PC := 2 p_{\g}(Z) -1$ where $p_{\g}(Z)$ is the probability of correctly guessing $Z$. Non-trivial predictability is typically obtained by choosing $\bs_1$ to be \textit{asymmetric}. In such situations, the following bound holds \cite{PhysRevD.19.473, Greenberger1988391}:
\begin{equation}
\label{eqn1bb}
\PC^2+ \VC^2  \leq 1.
\end{equation}
This particularly simple example is a special case of Robertson's uncertainty relation involving standard deviations~\cite{PhysRevA.60.1874, DurrRempe2000, Busch20061, BosykEtAl2013}. However,~\cite{BosykEtAl2013} argues that \eqref{eqn1bb} is inequivalent to a family of EURs where the same (R\'{e}nyi) entropy is used for both uncertainty terms, hence one gets the impression that entropic uncertainty is different from wave-particle duality. On the other hand,~\cite{BosykEtAl2013} did not consider the EUR involving the min- and max-entropies. For some probability distribution $P = \{p_j\}$, the unconditional min- and max-entropies are given by $H_{\min}(P)= - \log \max_j p_j$ and $H_{\max}(P)=  2 \log \sum_j \sqrt{p_j}$. We find that \eqref{eqn1bb} is equivalent to
\begin{equation}
\label{eqn1aa}
H_{\min}(Z) +\min_{W\in XY} H_{\max}(W) \geq 1,
\end{equation}
which is an EUR proved in the seminal paper by Maassen and Uffink \cite{MaassenUffink}, and corresponds to $E_1$ and $E_2$ in \eqref{eqnMainResult} being trivial. The entropies in \eqref{eqn1aa} are evaluated for the state at any time while the photon is inside the interferometer. It is straightforward to see that $H_{\min}(Z) = -\log \frac{1+\PC}{2}$ and in the Methods we prove that 
\begin{equation}
\label{eqnVisMaxRelation}
\min_{W\in XY} H_{\max}(W) = \log (1+\sqrt{1-\VC^2}).
\end{equation}
Plugging these relations into \eqref{eqn1aa} gives \eqref{eqn1bb}.

\textit{$\DC$-$ \VC$ relation.---}Let us move on to a more general and more interesting scenario where, in addition to prior which-path knowledge, one may obtain further knowledge \textit{during} the experiment due to the interaction of the photon with some environment $\En$, which may act as a which-way detector. Most generally the interaction is given by a completely positive trace preserving (CPTP) map $\EC$, with the input system being $Q$ at time $t_1$ and output systems being $Q$ and $\En$ at time $t_2$, see Fig.~\ref{fgrMZ}. The final state is $\rho^{(2)}_{Q \En} = \EC(\rho_Q^{(1)})$, where the superscripts $(1)$ and $(2)$ indicate the states at times $t_1$ and $t_2$. We do not require $\EC$ to have any special form in order to derive our WPDR, so our treatment is general.

The path distinguishability is defined by $\DC := 2 p_{\g}(Z| \En) -1$, where $p_{\g}(Z| \En)$ is the probability for correctly guessing the photon's path $Z$ at time $t_2$ given that the experimenter performs the optimally helpful measurement on $\En$. We find that \eqref{eqn1} is equivalent to 
\begin{equation}
\label{eqn1ent}
H_{\min}(Z| \En)+\min_{W\in XY} H_{\max}(W) \geq 1,
\end{equation}
where the entropy terms are evaluated for the state $\rho^{(2)}_{Q \En}$, which corresponds to $E_1= \En$ and $E_2$ being trivial in \eqref{eqnMainResult}. First, it is obvious from the operational meaning of the conditional min-entropy \eqref{eq:minabsdef} that we have $H_{\min}(Z| \En) =-\log \frac{1+\DC}{2}$, and second we use our result \eqref{eqnVisMaxRelation} to rewrite \eqref{eqn1ent} as \eqref{eqn1}. As emphasized in \cite{EnglertPRL1996}, we note that \eqref{eqn1} and its entropic form \eqref{eqn1ent} do not require $\bs_1$ to be symmetric. Hence $\DC$ accounts for both the prior $Z$ knowledge associated with the asymmetry of $\bs_1$ as well as the $Z$ information gained from $\En$.

\begin{figure}[t]
\begin{center}
\includegraphics[scale=0.94]{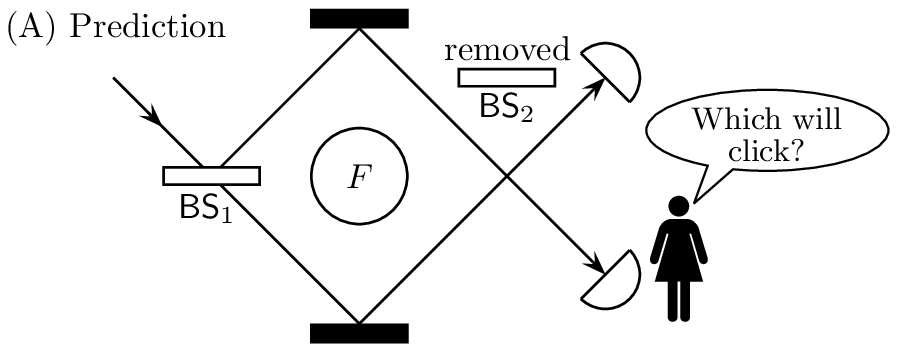}

\vspace{9pt}
\includegraphics[scale=0.94]{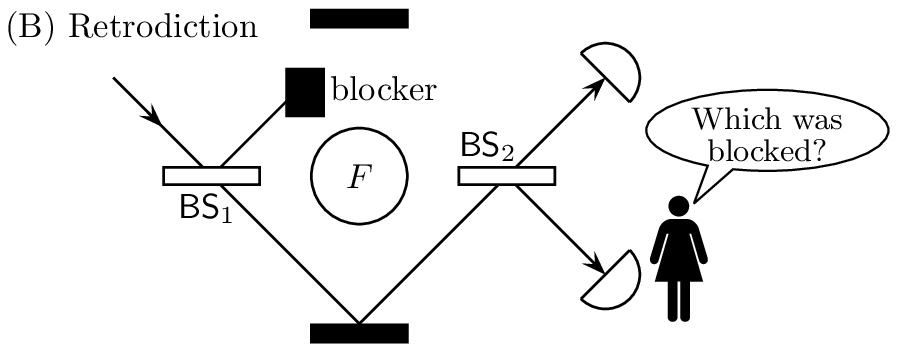}

\caption{%
Path prediction versus path retrodiction, in the MZI. (A) In the predictive scenario, the second beam splitter is removed and Alice tries to guess which detector will click. (B) In the retrodictive scenario, a blocker is randomly inserted into one of the interferometer arms and Alice tries to guess which arm was blocked (given the knowledge of which detector clicked).  \label{fgrPR}}
\end{center}
\end{figure}

\textit{Preparation vs.\ measurement uncertainty.---}The above analysis shows that \eqref{eqn1} and \eqref{eqn1bb} correspond to applying the \textit{preparation uncertainty relation} at time $t_2$ (just before the photon reaches $\bs_2$). Preparation uncertainty restricts one's ability to predict the outcomes of \textit{future} measurements of complementary observables. Thus, to experimentally measure $\PC$ or more generally $\DC$, the experimenter \textit{removes} $\bs_2$ and sees how well he/she can guess which detector clicks, see Fig.~\ref{fgrPR}A. Of course, to then measure $\VC$, the experimenter reinserts $\bs_2$ to close the interferometer. We emphasize that this procedure falls into the general framework of \textit{preparation uncertainty}.

On the other hand, uncertainty relations can be applied in a conceptually different way. Instead of two complementary output measurements and a fixed input state, consider a fixed output measurement and two complementary sets of input states. Namely consider the input ensembles from \eqref{eqnZ11} and \eqref{eqnW11}, now labeled as $Z_i = \{\ket{0}, \ket{1}\}$ and $W_i = \{\ket{w_{\pm}} \}$, where $i$ stands for ``input'', to indicate the physical scenario of a sender inputting states into a channel. Imagine this as a retrodictive guessing game, where Bob controls the input and Alice has control over both $\En$ and the detectors. Bob chooses one of the ensembles and flips a coin to determine which state from the ensemble he will send, and Alice's goal is to guess Bob's coin flip outcome. Assuming $\bs_1$ is 50/50, the two $Z_i$ states are generated by Bob blocking the opposite arm of the interferometer, as in Fig.~\ref{fgrPR}B, while the $W_i$ states are generated by applying a phase (either $\phi_0$ or $\phi_0+\pi$) to the lower arm.

It may not be common knowledge that this scenario leads to a different class of WPDRs, therefore we illustrate the difference in Fig.~\ref{fgrPR}. For clarity, we refer to $\DC$ introduced above as output distinguishability, whereas in the present scenario we use the symbol $\DC_i$ and call this quantity input distinguishability, defined by
\begin{equation}
\label{eqnVinput}
\DC_i := 2 p_{\g}(Z_i |\En )_{D_0}-1,
\end{equation}
where $p_{\g}(Z_i | \En )_{D_0}$ is Alice's probability to correctly guess Bob's $Z_i$ state given that she has access to $\En$ and she knows that detector $D_0$ clicked at the output. Likewise we define the notion of \textit{input} visibility $\VC_i$ via:
\begin{equation}
\label{eqnVinput}
\VC_i := \max_{W\in XY} [2 p_{\g}(W_i )_{D_0}-1]
\end{equation}
which quantifies how well Alice can determine $W_i$ given that she knows $D_0$ clicked.

Now the uncertainty principle says there is a tradeoff: if Alice can guess the $Z_i$ states well then she cannot guess the $W_i$ states well, and vice-versa. In other words, Alice's measurement apparatus, the apparatus to the right of the dashed line labeled $t_1$ in Fig.~\ref{fgrMZ}, cannot \textit{jointly measure} Bob's $Z$ and $W$ observables. EURs involving von Neumann entropy have previously been applied to the joint measurement scenario \cite{ColesEtAlPRA2011, PhysRevLett.112.050401}, we do the same for the min- and max-entropies to obtain (see Methods for details) 
\begin{equation}
\label{eqn1input}
\DC_i^2 + \VC_i^2 \leq 1,
\end{equation}
which can now be applied to a variety of situations.

\begin{figure}[t]
\begin{center}
\includegraphics[scale=0.95]{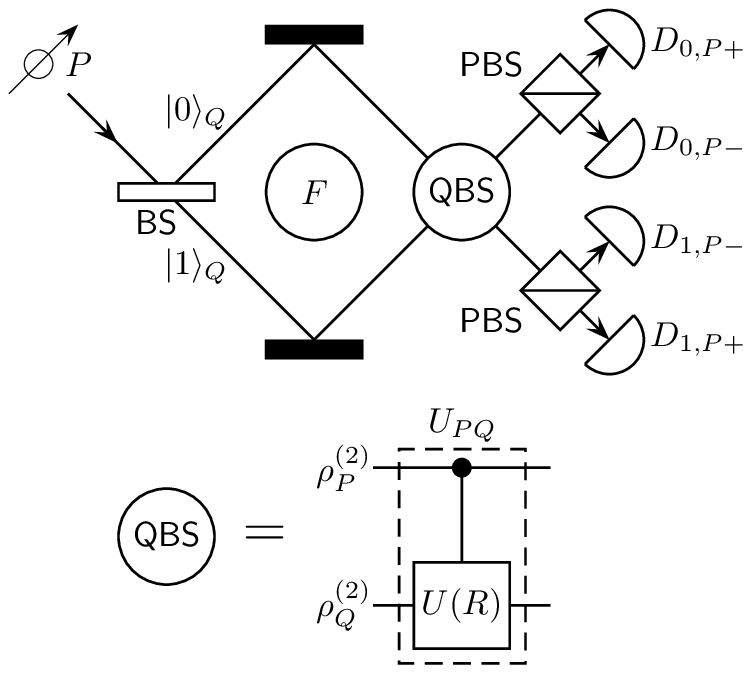}
\caption{%
In the quantum beam splitter ($\qbs$) scenario, the second beam splitter is in a superposition of ``absent'' and ``present'', as determined by the polarization state $\rho^{(2)}_P$ at time $t_2$. The $\qbs$ can be modelled as a controlled-unitary, $U_{PQ} = \dya{H}_P\ot \id_Q+\dya{V}_P\ot U(R)$, where $U(R)$ is the unitary on $Q$ associated with an asymmetric beam splitter with reflection probability $R$. Polarization-resolving detectors ($\pbs$ = polarizing beam splitter) on the output modes help to reveal the ``quantumness'' of the $\qbs$.\label{fgr3}}
\end{center}
\end{figure}

\textit{Quantum $\bs_2$.---}As an interesting application of \eqref{eqn1input}, we consider the scenario proposed in \cite{PhysRevLett.107.230406} and implemented in \cite{Kaiser02112012, Peruzzo02112012, TangEtAlPhysRevA.88.014103}, where the photon's polarization $P$ acts as a control system to determine whether or not $\bs_2$ appears in the photon's path and hence whether the interferometer is open or closed, see Fig.~\ref{fgr3}. Since $P$ can be prepared in an arbitrary input state $\rho^{(2)}_{P}$, such as a superposition, this effectively means that $\bs_2$ is a ``quantum beam splitter'', i.e., it can be in a quantum superposition of being absent or present. The interaction coupling $P$ to $Q$ is modelled as a controlled unitary as in Fig.~\ref{fgr3}. In this case the two visibilities are equivalent (see Methods)
\begin{equation}
\label{eqn1843532gen}
\VC_i = \VC = 2 | \kp | \sqrt{R(1-R)} \mte{V}{\rho^{(2)}_{P}}
\end{equation}
where we assume the dynamics are path-preserving, i.e., $\EC_{Q}(\dya{0}) = \dya{0}$ and $\EC_{Q}(\dya{1}) = \dya{1}$, where $\EC_{Q} = \Tr_{\En}\circ \EC$ is the reduced channel on $Q$, which implies that $\EC_{Q}(\dyad{0}{1}) =\kp \dyad{0}{1}$, i.e., off-diagonal elements get scaled by a complex number $\kp$ with $|\kp | \leq 1$. In \eqref{eqn1843532gen}, $\VC$ is evaluated for any pure state input $\rho^{(1)}_{Q}$ from the $XY$ plane of the Bloch sphere (e.g., $\ket{+}$). Now we apply \eqref{eqn1input} to this scenario and use \eqref{eqn1843532gen} to obtain:
\begin{equation}
\label{eqn2394352}
\DC_i^2 +\VC^2 \leq 1,
\end{equation}
which extends a recent result in~\cite{JiaEtAlCPB2014} to the case where $\En$ is non-trivial. This general treatment includes the special case where $\rho^{(2)}_{P} = \dya{V}$, corresponding to a closed interferometer with an asymmetric $\bs_2$. Ref.~\cite{PhysRevLett.100.220402} experimentally tested this special case. However,~\cite{PhysRevLett.100.220402} did not remark that their experiment actually tested a relation different from \eqref{eqn1}, namely they tested a special case of \eqref{eqn2394352}.

Similarly,~\cite{Kaiser02112012} tested \eqref{eqn2394352} rather than \eqref{eqn1}, but they allowed $\rho^{(2)}_{P}$ to be in a superposition. At first sight this seems to test the WPDR in the case of a quantum beam splitter, but it turns out that neither the visibility $\VC$ nor the distinguishability $\DC_i$ depends on the phase coherence in $\rho^{(2)}_{P}$ and hence the data could be simulated by a classical mixture of $\bs_2$ being absent or present. Nevertheless, our framework provides a WPDR that captures the coherence in $\rho^{(2)}_{P}$, by conditioning on the polarisation $P$ at the interferometer output (see Methods). For example, defining the polarization-enhanced distinguishability, $\DC_i^P:= 2 p_{\g}(Z_i |\En P)_{D_0}-1$, which corresponds to choosing $E_1 = FP$, we obtain the novel WPDR: 
\begin{equation}
\label{eqn23943521}
(\DC_i^P)^2 +\VC^2 \leq 1,
\end{equation}
which captures the beam splitter's coherence (see \footnoteremember{SuppFootnote}{See the Supplementary Information.} for elaboration) and could be tested with the setup in \cite{Kaiser02112012}.

\textit{Non-trivial $E_2$.---}The above examples use the environment solely to enhance the particle behavior. To give a corresponding example for wave behavior, i.e., where system $E_2$ in \eqref{eqnMainResult} is non-trivial, the main result of \cite{Englert2000337} is a WPDR for the case when the environment $\En$ is measured (after it has interacted with the quanton) and the resulting information is used to enhance the fringe visibility. This scenario is called \textit{quantum erasure} since the goal is to erase the which-path information stored in the environment, to recover full visibility. This falls under our framework by taking $E_2$ to be the classical output of the measurement on the environment. For elaboration see \footnoterecall{SuppFootnote}, where we also cast the main results of Refs.~\cite{PhysRevA.85.054101} and \cite{Banaszek:2013fk} within our framework.

\textit{Conclusions.---}We have unified the wave-particle duality principle and the entropic uncertainty principle, showing that WPDRs are EURs in disguise. We leave it for future work to extend this connection to multiple interference pathways \cite{EnglertIJQI2008}. The framework presented here can be applied universally to binary interferometers.  Our framework makes it clear how to formulate novel WPDRs by simply applying known EURs to novel interferometer models, and these new WPDRs will likely inspire new interferometry experiments. We note that all of our relations also hold if one replaces both min- and max-entropy with the well-known von Neumann entropy. Alternatively, one can use smooth entropies \cite{TomRen2010, TLGR}, and the resulting smooth WPDRs may find application in the security analysis of interferometric quantum key distribution \cite{PhysRevLett.69.1293}, which often exploits the Franson setup (Fig.~\ref{fgrFran}).

\section*{METHODS}

We emphasize that our treatment, in what follows, will be for a generic binary interferometer. We will first discuss our general treatment, then we will specialize to the predictive and retrodictive scenarios (see Fig.~\ref{fgrPR}).

\textit{Origin of general WPDR.---}It is known that the min- and max-entropies satisfy the uncertainty relation \cite{TomRen2010}:
\begin{equation}
\label{eqn1genminmaxUR1}
H_{\min}(Z | E_1)+H_{\max}(W | E_2)\geq 1,
\end{equation}
for any tripartite state $\rho_{A E_1 E_2}$ where $A$ is a qubit and $Z$ and $W$ are mutually unbiased bases on $A$. Noting that the which-path and which-path observables in \eqref{eqnZ11} and \eqref{eqnW11} are mutually unbiased (for all $\phi_0$ in \eqref{eqnW11}, i.e., for all $W$ in the $XY$ plane) gives our general WPDR in \eqref{eqnMainResult}.

\textit{Complementary guessing game.---}The operational interpretation of \eqref{eqnMainResult} in terms of the complementary guessing game described, e.g., in Figs.~\ref{fgrMZ}-\ref{fgrDS} can be seen clearly as follows. While the min-entropy is related to the guessing probability via \eqref{eq:minabsdef}, we establish a similar relation for the max-entropy. First we prove \footnoterecall{SuppFootnote} that, for a general classical-quantum state $\rho_{XB} = \sum_j \dya{j}\ot \sg_{B}^j$ where $X$ is binary,
\begin{equation}
\label{eqnmaxdef111}
H_{\max}(X|B) = \log \Big(1+2 \norm[\big]{\sqrt{\sg_{B}^0}  \sqrt{\sg_{B}^1}}_1 \Big),
\end{equation}
where the 1-norm is $\| M \|_1 = \Tr \sqrt{M\ad M}$. Next we show \footnoterecall{SuppFootnote}, for any positive semi-definite operators $M$ and $N$, 
\begin{equation}
\label{eqntrace-norm-fidelitymain}
\norm{M - N}_{1}^{2} + 4 \norm{\sqrt{M} \sqrt{N}}_{1}^{2} \leq (\Tr M + \Tr N)^{2}.
\end{equation}
Combining \eqref{eqntrace-norm-fidelitymain} with \eqref{eqnmaxdef111}, and using the well-known formula $\norm{\sg_{B}^0 - \sg_{B}^1}_{1} = 2 p_{\g}(X|B)-1$, gives
\begin{equation}
\label{eqnmaxlemma112}
H_{\max}(X|B) \leq \log \Big(1+\sqrt{1-(2 p_{\g}(X|B)-1)^2}  \Big).
\end{equation}
Now one can define generic measures of particle and wave behavior directly in terms of the guessing probabilities:
\begin{align}
\label{eqnsymmeasures1}
\DC_g&:= 2 p_{\g}(Z|E_1)-1,\\
\VC_g &:= \max_{W\in XY} [2 p_{\g}(W|E_2)-1]
\end{align}
for some arbitrary quantum systems $E_1$ and $E_2$, and rearrange \eqref{eqnMainResult} into the traditional form for WPDRs:
\begin{equation}
\label{eqnGenVis23434}
\DC_g^2+\VC_g^2 \leq 1.
\end{equation}
This operationally-motivated relation, which follows directly from \eqref{eqnMainResult}, clearly imposes a restriction on Alice's ability to win the complementary guessing game, since $\DC_g$ and $\VC_g$ are defined in terms of the winning probabilities. Below we show that $\VC_g$ becomes the fringe visibility when $E_2$ is discarded.

\textit{Predictive WPDRs.---}We now elaborate on our framework for deriving predictive WPDRs. Let us denote the quanton's spatial degree of freedom as $S$, which includes the previously mentioned $Q$ as a subspace. At time $t_2$ (see, e.g., Fig.~\ref{fgrMZ}) - the time just before a phase $\phi$ is applied and the interferometer is closed - $S$ and its environment $E$ are in some state $\rho^{(2)}_{SE}$, where again $E=E_1 E_2$ is a generic bipartite system. The preparation is arbitrary, i.e., we need not specify what happened at earlier times, such as what the system's state was at time $t_1$ (prior to the interaction between $S$ and $E$). While in general a binary interferometer may have more than two paths, all but two of these are non-interfering (by definition), hence we only consider the two-dimensional subspace associated with the two which-path states of interest, denoted $\ket{0}$ and $\ket{1}$. This subspace, defined by the projector $\Pi :=\dya{0}+\dya{1}$, can be thought of as an effective qubit system $Q$. (Note that $Q=S$ in the MZI.) Without loss of generality, we project the state $\rho^{(2)}_{SE}$ onto this subspace and denote the resulting (renormalized) state as 
\begin{equation}
\label{eqnrho2qe2342}
\rho^{(2)}_{QE} = (\Pi \ot \id)\rho^{(2)}_{SE} (\Pi \ot \id) / \Tr(\Pi \rho^{(2)}_{S}).
\end{equation}
Experimentally this corresponds to post-selecting on the interfering portion of the data. To derive predictive WPDRs, we apply \eqref{eqnMainResult} to the state $\rho^{(2)}_{QE}$ in \eqref{eqnrho2qe2342}, where we associate the subsystems $E_1$ and $E_2$ of $E$ with the particle and wave terms respectively. 

For example this approach gives the WPDRs discussed in \cite{EnglertPRL1996}, Eqs.~\eqref{eqn1} and \eqref{eqn1bb}. To show this we must prove \eqref{eqnVisMaxRelation}, which relates our entropic measure of wave behavior in \eqref{eq:lackofwave} to fringe visibility, and we now do this for generic binary interferometers. We remark that one can take \eqref{eqnfringevis1} as a generic definition for fringe visibility, where the label $D_0$ is arbitrary, i.e., it corresponds to some arbitrary detector. For generic binary interferometers, there is a phase shift $\phi$ applied just after time $t_2$, as depicted in Fig.~\ref{fgrMZ}. Let $U_{\phi} = \dya{0}+e^{i\phi}\dya{1}$ denote the unitary associated with this phase shift, and note that we only need to specify the action of $U_{\phi}$ on the $Q$ subspace since the state $\rho^{(2)}_{QE}$ lives in this subspace. 

Finally the quanton is detected somewhere, i.e., system $S$ is measured and a detector $D_0$ clicks. This measurement is a positive operator valued measure (POVM) $C=\{C_0, C_1,...\}$ on the larger space, system $S$ rather than the subspace $Q$ (e.g., think of the double slit case, where the detection screen performs a position measurement on $S$). We associate the POVM element $C_0$ with the event of detector $D_0$ clicking. To prove \eqref{eqnVisMaxRelation}, we need to restrict the form of $C_0$. We show that \eqref{eqnVisMaxRelation} holds so long as $C_0$ is \textit{unbiased} with respect to the which-path basis $Z$ on the subspace $Q$. Fortunately this condition is satisfied for all three types of interferometers in Figs.~\ref{fgrMZ},~\ref{fgrFran}, and ~\ref{fgrDS}. More precisely, it is satisfied for the MZI provided $\bs_2$ is 50/50, for the Franson case provided \text{both} $\bs_2$ (the second beam splitters in Fig.~\ref{fgrFran}) are 50/50, and for the double slit for some limiting choice of experimental parameters such as large $L$ in Fig.~\ref{fgrDS}. We now state a general lemma that applies to all of these interferometers.

\begin{lemma}
\label{lem1}
Consider a binary interferometer where $\tilde{C}_0:= \Pi C_0 \Pi$ denotes the projection of POVM element $C_0$ onto the interfering subspace ($Q$). Suppose $\tilde{C}_0$ is proportional to a projector projecting onto a state from the $XY$ plane of the Bloch sphere of $Q$, i.e.,
\begin{equation}
\label{eqnC0condition}
\tilde{C}_0 = q \dya{w_{+}}
\end{equation}
for some $0< q \leq 1$, where $\ket{w_{+}}$ is given by \eqref{eqnW11} for some arbitrary phase $\phi_0$. Then it follows that
\begin{equation}
\label{eqnVisMaxRelationMeth}
\min_{W\in XY} H_{\max}(W) = \log (1+\sqrt{1-\VC^2}),
\end{equation}
where $\VC$ is given by \eqref{eqnfringevis1}, and $H_{\max}(W)$ is evaluated for the state $\rho_{Q}^{(2)}=\Tr_E(\rho_{QE}^{(2)})$.
\end{lemma}
\begin{proof}
In what follows it should be understood that probabilities and expectation values are evaluated for the state $\rho_Q^{(2)}$. Suppose that $\widetilde{W}$ is optimal in the sense that $\max_{W\in XY} \Pr (w_{+}) = \Pr (\widetilde{w}_+)$ where $\Pr (w_{\pm}) := \mte{w_{\pm}}{\rho^{(2)}_Q}$. Then we have
\begin{align}
\min_{W\in XY} H_{\max}(W)&=\log \Big(1+\sqrt{1-\ave{\sg_{\widetilde{W}}   }^2}\Big) 
\end{align}
where we denote Pauli operators by $\sg_{W}:= \dya{w_+}-\dya{w_-}$, and $\ave{\sg_{\widetilde{W}} }= \Pr (\widetilde{w}_+)-\Pr (\widetilde{w}_-)$.

The probability for $D_0$ to click is
\begin{equation}
\label{eqnpD0click}
p^{D_0} = \Tr(C_0 U_{\phi} \rho^{(2)}_Q U_{\phi}\ad) = \Tr(U_{\phi}\ad \tilde{C}_0 U_{\phi} \rho^{(2)}_Q)
\end{equation}
and maximising this over $\phi$ gives
\begin{equation}
\label{eqnpD0clickmax}
p^{D_0}_{\max} = q \max_{W\in XY} \Pr (w_+) = q \Pr (\widetilde{w}_+).
\end{equation}
Now, due to the geometry of the Bloch sphere, we have $p^{D_0}_{\min} = \Pr (\widetilde{w}_-)$. Thus, $p^{D_0}_{\max}+ p^{D_0}_{\min} = q$ and $p^{D_0}_{\max} - p^{D_0}_{\min} = q \ave{\sg_{\widetilde{W}}   }$. This gives $\VC = \ave{\sg_{\widetilde{W}}   }$, completing the proof.
\end{proof}

\textit{Retrodictive WPDRs.---}While we saw that the predictive approach allowed for any preparation but required complementary output measurements, the opposite is true in the retrodictive case, i.e., the form of the output measurement is arbitrary while we require complementary preparations. The input ensembles $Z_i = \{\ket{0},\ket{1}\}$ and $W_i = \{\ket{w_+},\ket{w_-}\}$ can be generated by performing the relevant measurements on a reference qubit $Q'$ that is initially entangled to the quanton $S$. Associating state ensembles with measurements on a reference system is a useful trick, e.g., for deriving \eqref{eqn1input}. Thus, at time $t_1$ (just after the quanton enters the interferometer, see Fig.~\ref{fgrMZ}) we introduce a qubit $Q'$ that is maximally entangled to the interfering subspace ($Q$) of $S$, denoted by the state $\rhob^{(1)}_{Q'S} = \dya{\Phi}$ with $\ket{\Phi} = (\ket{00}+\ket{11})/\sqrt{2}$. The dynamics after time $t_1$ is modelled as a quantum operation $\AC$, defined in \cite{NieChu00} as a completely positive, trace non-increasing map, that maps $S \to E_1E_2$. The output of $\AC$ does not contain $S$ because the quanton is eventually detected by a detector, at which point we no longer need a quantum description the quanton's spatial degree of freedom; we only care where it was detected. The map $\AC$ corresponds to a particular detection event; for concreteness say that detector $D_0$ clicking is the associated event. The probability for this event is the trace of the state after the action of $\AC$, and renormalizing gives the final state
\begin{equation}
\label{eqnrhob2342main}
\rhob^{D_0}_{Q' E_1E_2}:=\frac{  (\IC \ot \AC)(\rhob^{(1)}_{Q'S}) }{ \Tr [ (\IC \ot \AC)(\rhob^{(1)}_{Q'S}) ] }.
\end{equation}
Our framework applies the uncertainty relation \eqref{eqnMainResult} to the state $\rhob^{D_0}_{Q' E_1E_2}$ to derive retrodictive WPDRs.

For example, this covers the scenario from the Discussion where $\AC$ involves two sequential steps. First $S$ interacts with an environment $F$ inside the interferometer between times $t_1$ and $t_2$, which corresponds to a channel $\EC$ mapping $S$ to $S \En$. Second, the quanton is detected at the interferometer output, say at detector $D_0$, modelled as a map $\BC(\cdot) = \Tr_S [C_0(\cdot)]$ acting on $S$, where $C_0$ is the POVM element associated with detector $D_0$ clicking. Hence we choose $\AC = \BC \circ \EC$. Applying \eqref{eqnMainResult} to this case while choosing $E_1 = F$ and $E_2$ to be trivial gives
\begin{equation}
\label{eqnJMeurmain1}
H_{\min}(Z | \En )_{\rhob}+\min_{W\in XY}H_{\max}(W)_{\rhob}\geq 1,
\end{equation}
where the subscript $\rhob$ means evaluating on the state in \eqref{eqnrhob2342main}. Note that measuring $Z$ on system $Q'$ corresponds to sending the states $\{\ket{0},\ket{1}\}$ with equal probability through the interferometer, and similarly for $W$ (with an inconsequential complication of taking the transpose of the $W$ basis states). Realizing this, the first and second terms in \eqref{eqnJMeurmain1} map onto $\DC_i$ and $\VC_i$ respectively: 
\begin{align}
\label{eqnDiVi} H_{\min}(Z | \En )_{\rhob} &= 1-\log (1+\DC_i ), \\
\label{eqnDiVi222} \min_{W\in XY}H_{\max}(W)_{\rhob} &= \log (1+\sqrt{1- \VC_i^2 }).
\end{align}
Hence \eqref{eqnJMeurmain1} becomes \eqref{eqn1input}.

It remains to show that $\VC_i$ appearing in \eqref{eqn1input} can be replaced by $\VC$ for many cases of interest, such as the $\qbs$ case. We do this in the following lemma, where the proof is given in \footnoterecall{SuppFootnote} and is similar to the proof of Lemma~\ref{lem1}.

\begin{lemma}
\label{lem22main}
Consider any binary interferometer with an unbiased input, i.e., where the state at time $t_1$ is unbiased with respect to the which-path basis (of the form $\ket{\psi_Q^{(1)}} =(\ket{0}+e^{i\phi}\ket{1})/\sqrt{2}$). Let $\EC_S = \Tr_F \circ \EC$ be the channel describing the quanton's interaction with $F$ inside the interferometer, and let $\GC(\cdot) = \Pi (\cdot)\Pi$ be the map that projects onto the subspace $\Pi $. Suppose $\EC_S$ is path-preserving, i.e., $\EC_S(\dya{0})=\dya{0}$ and $\EC_S(\dya{1})=\dya{1}$ and furthermore suppose $\EC_S$ commutes with $\GC$. Then
\begin{equation}
\label{eqnVisMaxRelationMeth2}
\min_{W\in XY} H_{\max}(W)_{\rhob} = \log (1+\sqrt{1-\VC^2}),
\end{equation}
where $H_{\max}(W)$ is evaluated for the state $\rhob_{Q'}^{D_0}$.
\end{lemma}

\textit{$\qbs$ example.---}Finally, we treat the quantum beam splitter shown in Fig.~\ref{fgr3}. (Note that $S=Q$ in the MZI.) This setup involves first a quantum channel $\EC$ that describes the interaction of $S$ with an environment $F$ between times $t_1$ and $t_2$, followed by another channel associated with the $\qbs$ that interacts $S$ with the polarization $P$, followed by a post-selected detection at $D_0$. Together these three steps form a quantum operation $\AC$ that maps $S \to \En P$, and hence this falls under our retrodictive framework.

To prove \eqref{eqn23943521} we apply \eqref{eqnMainResult} to the state in \eqref{eqnrhob2342main} while choosing $E_1 = \En P$ and $E_2$ to be trivial, giving
\begin{equation}
\label{eqnJMeurQBS}
H_{\min}(Z| \En P)_{\rhob}+\min_{W\in XY}H_{\max}(W)_{\rhob}\geq 1.
\end{equation}
We then use relations analogous to those in \eqref{eqnDiVi} and \eqref{eqnDiVi222}, where the former relation now involves conditioning also on the polarisation $P$. Finally, we note that Lemma~\ref{lem22main} applies to the $\qbs$ case.

\section*{ACKNOWLEDGEMENTS}

We thank B.\ Englert and S.\ Tanzilli for helpful correspondence, and acknowledge helpful discussions with M.\ Woods, M.\ Tomamichel, C.~J.\ Kwong, and L.~C.\ Kwek. We acknowledge funding from the Ministry of Education (MOE) and National Research Foundation Singapore, as well as MOE Tier 3 Grant ``Random numbers from quantum processes'' (MOE2012-T3-1-009).

\clearpage

\onecolumngrid

\renewcommand{\thetheorem}{S\arabic{theorem}}
\renewcommand{\thefigure}{S\arabic{figure}}
\renewcommand{\theequation}{S\arabic{equation}}

\section*{\large Supplementary Information}

\starttocentries

\tableofcontents

\section{Introduction}

In this Supplementary Information, we elaborate on the technical details justifying our claims. Furthermore, to emphasize the universality of our framework, we provide additional results showing that other WPDRs appearing in the literature can be phrased within our framework.

In what follows, Sec.~\ref{sct1A} proves the relation between the max-entropy and the guessing probability given in the Methods section. In Sec.~\ref{scthybrid} we extend our framework to scenarios where both the preparation (at the interferometer input) and the measurement outcome (at the interferometer output) may provide information about the quanton's path (inside the interferometer). Such situations are a ``hybrid'' of the predictive and retrodictive scenarios discussed in the main text. This extension allows us to reinterpret a result in Ref.~\cite{PhysRevA.85.054101} for asymmetric beam splitters as an entropic uncertainty relation and hence show that it falls under our framework. Subsection~\ref{sctProofLem2} proves Lemma~2 from the main text, or in fact, proves a generalized version of this lemma that holds for these ``hybrid'' scenarios. In Sec.~\ref{sct4A} we consider WPDRs involving enhanced visibilty. In particular we show that the results of Ref.~\cite{Englert2000337} for quantum erasure and Ref.~\cite{Banaszek:2013fk} for polarization dynamics can both be viewed as entropic uncertainty relations, where the visibility term is enhanced by conditioning on additional information. Finally in Sec.~\ref{sct2A} we elaborate on our novel WPDR for the quantum beam splitter, demonstrating that it captures the coherence in the $\qbs$ and discussing how the polarization-enhanced distinguishability can be measured.

\section{Relating max-entropy to guessing probability}\label{sct1A}

Here we relate the max-entropy to the guessing probability. The main lemma that we prove explicitly solves for the max-entropy of a classical-quantum (cq) state where the classical register is binary, which to our knowledge is a new result. An arbitrary cq state where the classical register is binary, which is the only case relevant to our analysis, can be written as $\rho_{XB} = \ketbraq{0} \otimes \sigma_{0} + \ketbraq{1} \otimes \sigma_{1}$, where $\sigma_{0}$ and $\sigma_{1}$ are subnormalized states that satisfy $\Tr \sigma_{0} + \Tr \sigma_{1} = 1$. In this case, the optimal guessing probability takes the form
\begin{equation}
p_{\g}(X | B) := \max_{M_{0}, M_{1}} \Tr ( M_{0} \sigma_{0}) + \Tr( M_{1} \sigma_{1} ),
\end{equation}
where the maximization is taken over all POVMs on subsystem $B$, namely operators $M_{0}, M_{1} \geq 0$ such that $M_{0} + M_{1} =~\id$. Since this is exactly the state discrimination problem solved by Helstrom~\cite{helstrom76} we have
\begin{equation}
\label{eqnpguesstd1}p_{\g}(X | B) = \frac{1}{2} + \frac{1}{2} \norm{\sigma_{0} - \sigma_{1}}_{1}.
\end{equation}
Hence the guessing probability is related to the trace distance between the conditional states.

The formula for the max-entropy was given in Eq.~(6b) from the main text and was expressed in terms of the fidelity, which is defined as
\begin{equation}
F(\AAz , \AAw ) := \norm[\big]{\sqrt{\AAz } \sqrt{\AAw }}_{1},
\end{equation}
for two positive semi-definite operators $\AAz $ and $\AAw $. In our case of a cq state with a binary register, the formula given in Eq.~(6b) simplifies to \cite{KonRenSch09}
\begin{equation}
\label{eq:hmax-maximization}
H_{\textnormal{max}}(X | B) = 2 \log \max_{\rho} \big( F(\sigma_{0}, \rho) + F(\sigma_{1}, \rho) \big),
\end{equation}
where the maximization is taken over all normalized states on $B$. The following lemma derives the optimal value of this optimization problem.
\begin{lemma}
Let $\AAz , \AAw  \geq 0$ be positive semi-definite operators and let $\mathcal{S}$ be the set of positive semi-definite operators with unit trace. Then
\begin{equation}
\max_{\rho \in \mathcal{S}} \big( F(\AAz , \rho) + F(\AAw , \rho) \big) = \sqrt{ \Tr \AAz  + \Tr \AAw  + 2 F(\AAz , \AAw ) }.
\end{equation}
\end{lemma}
\begin{proof}
First, we show that the right-hand side constitutes a valid upper bound and then we give an explicit choice of $\rho$ that achieves it.

For arbitrary unitaries $U_{0}$ and $U_{1}$ let $X^{\dagger} = U_{0} \sqrt{\AAz } + U_{1} \sqrt{\AAw }$ and $Y = \sqrt{\rho}$. The Cauchy-Schwarz inequality, $\abs{\Tr (X^{\dagger} Y)}^{2} \leq \Tr (X^{\dagger} X) \cdot \Tr (Y^{\dagger} Y)$, implies that
\begin{equation}
\label{eq:cauchy-schwarz}
\abs[\Big]{\Tr \Big( \big( U_{0} \sqrt{\AAz } + U_{1} \sqrt{\AAw } \big) \sqrt{\rho} \Big)}^{2} \leq \Tr(X^{\dagger} X) \cdot \Tr \rho = \Tr(X^{\dagger} X).
\end{equation}
Since
\begin{equation}
X^{\dagger} X = U_{0} \AAz  U_{0}^{\dagger} + U_{1} \AAw  U_{1}^{\dagger} + U_{0} \sqrt{\AAz } \sqrt{\AAw } U_{1}^{\dagger} + U_{1} \sqrt{\AAw } \sqrt{\AAz } U_{0}^{\dagger}
\end{equation}
we have
\begin{align}
\Tr (X^{\dagger} X) &= \Tr \AAz  + \Tr \AAw  + \Tr \big( U_{0} \sqrt{\AAz } \sqrt{\AAw } U_{1}^{\dagger} + U_{1} \sqrt{\AAw } \sqrt{\AAz } U_{0}^{\dagger} \big)\\
&\leq \Tr \AAz  + \Tr \AAw  + \norm[\big]{U_{0} \sqrt{\AAz } \sqrt{\AAw } U_{1}^{\dagger} + U_{1} \sqrt{\AAw } \sqrt{\AAz } U_{0}^{\dagger}}_{1}\\
&\leq \Tr \AAz  + \Tr \AAw  + 2 \norm[\big]{\sqrt{\AAz } \sqrt{\AAw }}_{1} = \Tr \AAz  + \Tr \AAw  + 2 F(\AAz , \AAw ),
\end{align}
where we have used the fact that for Hermitian matrices $\Tr \; T \leq \norm{T}_{1}$ followed by the triangle inequality for the $1$-norm. Note that this bound is valid for all unitaries $U_{0}$ and $U_{1}$.

Let $L$ be a linear operator and let $L = U_{0} S U_{1}$ be its singular value decomposition. Clearly, $\norm{L}_{1} = \Tr (V L)$ for $V = U_{1}^{\dagger} U_{0}^{\dagger}$. Therefore, for every pair of positive semi-definite operators $A$ and $B$ there exists a unitary $V$ such that $F(A, B) = \norm{\sqrt{A} \sqrt{B}}_{1} = \Tr (V \sqrt{A} \sqrt{B})$. Let us choose unitaries $V_{0}$ and $V_{1}$ such that
\begin{equation}
F(\AAz , \rho) = \Tr \big( V_{0} \sqrt{\AAz } \sqrt{\rho} \big) \nbox{and} F(\AAw , \rho) = \Tr \big( V_{1} \sqrt{\AAw } \sqrt{\rho} \big).
\end{equation}
Adding these two terms together gives
\begin{equation}
F(\AAz , \rho) + F(\AAw , \rho) = \Tr \big( V_{0} \sqrt{\AAz } \sqrt{\rho} \big) + \Tr \big( V_{1} \sqrt{\AAw } \sqrt{\rho} \big) = \Tr \Big( \big( V_{0} \sqrt{\AAz } + V_{1} \sqrt{\AAw } \big) \sqrt{\rho} \Big).
\end{equation}
Since for this particular choice of unitaries the quantity on the right-hand side is real and positive we can apply~\eqref{eq:cauchy-schwarz} to obtain
\begin{equation}
\label{eq:upper-bound}
F(\AAz , \rho) + F(\AAw , \rho) = \abs[\Big]{\Tr \Big( \big( V_{0} \sqrt{\AAz } + V_{1} \sqrt{\AAw } \big) \sqrt{\rho} \Big)} \leq \sqrt{\Tr \AAz  + \Tr \AAw  + 2 F(\AAz , \AAw )}.
\end{equation}

Now, we simply need to provide a state $\rho$ that saturates this inequality. Taking advantage of the singular value decomposition of $\sqrt{\AAz } \sqrt{\AAw } = U_{0} S U_{1}$ (where $S$ is a diagonal matrix of real, non-negative numbers and $\Tr S = \norm{\sqrt{M} \sqrt{N}}_{1}$) we define $V = U_{1}^{\dagger} U_{0}^{\dagger}$ and
\begin{equation}
K = \AAz  + \AAw  + \sqrt{\AAz } V^{\dagger} \sqrt{\AAw } + \sqrt{\AAw } V \sqrt{\AAz }.
\end{equation}
Note that $K \geq 0$ since $K = \LLL^{\dagger} \LLL$ for $\LLL = \sqrt{\AAz } + V^{\dagger} \sqrt{\AAw }$. It is easy to verify that
\begin{equation}
\Tr \big( \sqrt{\AAz } V^{\dagger} \sqrt{\AAw } \big) = \Tr \big( \sqrt{\AAw } V \sqrt{\AAz } \big) = \Tr S = \norm{ \sqrt{\AAz } \sqrt{\AAw } }_{1} = F(\AAz , \AAw ),
\end{equation}
which implies that
\begin{equation}
\Tr K = \Tr \AAz  + \Tr \AAw  + 2 F(\AAz , \AAw )
\end{equation}
To calculate $F(\AAz , K) = \norm{\sqrt{\AAz } \sqrt{K}}_{1} = \Tr \sqrt{\sqrt{\AAz } K \sqrt{\AAz }}$ note that
\begin{align}
\sqrt{\AAz } K \sqrt{\AAz } &= \AAz ^{2} + \sqrt{\AAz } \AAw  \sqrt{\AAz } + \AAz  V^{\dagger} \sqrt{\AAw } \sqrt{\AAz } + \sqrt{\AAz } \sqrt{\AAw } V \AAz \\
&= \AAz ^{2} + U_{0} S^{2} U_{0}^{\dagger} + \AAz  U_{0} S U_{0}^{\dagger} + U_{0} S U_{0}^{\dagger} \AAz \\
&= \big( \AAz  + U_{0} S U_{0}^{\dagger} \big)^{2}.
\end{align}
Therefore, $F(\AAz , K) = \Tr \AAz  + F(\AAz , \AAw )$ and similarly $F(\AAw , K) = \Tr \AAw  + F(\AAz , \AAw )$. Since
\begin{equation}
F(\alpha A, B) = \norm{ \sqrt{\alpha A} \sqrt{B} }_{1} = \sqrt{\alpha} \norm{ \sqrt{A} \sqrt{B} }_{1} = \sqrt{\alpha} F(A, B)
\end{equation}
we can define $\rho = K / \Tr K$ which satisfies
\begin{equation}
F(\AAz , \rho) + F(\AAw , \rho) = \frac{F(\AAz , K) + F(\AAw , K)}{\sqrt{\Tr K}} = \sqrt{\Tr \AAz  + \Tr \AAw  + 2 F(\AAz , \AAw )}
\end{equation}
and saturates the bound~\eqref{eq:upper-bound}.
\end{proof}

Now taking the above lemma and setting $\AAz  = \sigma_{0}$ and $\AAw  = \sigma_{1}$ allows us to solve the maximization in~\eqref{eq:hmax-maximization}. We obtain the following result.

\begin{lemma}
\label{lem:max-fidelity}
For any cq state $\rho_{XB} = \ketbraq{0} \otimes \sigma_{0} + \ketbraq{1} \otimes \sigma_{1}$ where $X$ is binary,
\begin{equation}
\label{eqnmaxentfid1}
H_{\textnormal{max}}(X | B) = \log \big( 1 + 2 F(\sigma_{0}, \sigma_{1}) \big).
\end{equation}
\end{lemma}

Finally we relate the fidelity in \eqref{eqnmaxentfid1} to the trace distance in the guessing probability \eqref{eqnpguesstd1} with the following lemma.
\begin{lemma}
\label{lem:trace-norm-fidelity}
Let $\AAz , \AAw \geq 0$ be two positive semi-definite operators. Then we have
\begin{equation}
\label{eqntrace-norm-fidelity}
\norm{\AAz - \AAw}_{1}^{2} + 4 \norm{\sqrt{\AAz}  \sqrt{\AAw}}_{1}^{2} \leq (\Tr \AAz + \Tr \AAw)^{2}.
\end{equation}
\end{lemma}
\begin{proof}
Let $\ket{\Omega} = \sum_{k} \ket{k} \ket{k}$ and consider $\ket{\psi_{\Az}} = (\sqrt{\AAz}  \otimes \id) \ket{\Omega}$ and $\ket{\psi_{\Aw}} = (\sqrt{\AAw} \otimes U ) \ket{\Omega}$, where $U$ is a unitary. It is easy to verify that
\begin{equation}
\braket{\psi_{\Az}}{\psi_{\Aw}} = \bramatketq{\Omega}{\sqrt{\AAz}  \sqrt{\AAw} \otimes U} = \Tr (\sqrt{\AAz}  \sqrt{\AAw} U^{T}),
\end{equation}
where $^T$ denotes the transpose in the standard basis. By choosing $U^T = U_{2}\ad U_{1}\ad$, where $U_{1}$ and $U_{2}$ come from the singular value decomposition of $\sqrt{\AAz}  \sqrt{\AAw} = U_{1} S U_{2}$ we obtain $\braket{\psi_{\Az}}{\psi_{\Aw}} = \Tr S = \norm{\sqrt{\AAz}  \sqrt{\AAw}}_{1}$. Let $\Tr_{2}$ denote partial trace over the second subsystem. It is easy to check that $\Tr_{2} \dya{\psi_{\Az}} = \Az$ and $\Tr_{2} \dya{\psi_{\Aw}} = \Aw$. Since the trace norm is non-increasing under the partial trace we have
\begin{equation}
\norm{\AAz - \AAw}_{1} \leq \norm[\Big]{\dya{\psi_{\Az}} - \dya{\psi_{\Aw}} }_{1}.
\end{equation}
The rank of the Hermitian matrix $\HHH = \dya{\psi_{\Az}} - \dya{\psi_{\Aw}}$ is at most 2 and let us denote the non-zero eigenvalues by $\lambda_{1}$ and $\lambda_{2}$. It is easy to verify that
\begin{align}
\lambda_{1} + \lambda_{2} &= \Tr \HHH = \braketq{\psi_{\Az}} - \braketq{\psi_{\Aw}},\\
\lambda_{1}^{2} + \lambda_{2}^{2} &= \Tr \HHH^{2} = \big( \braketq{\psi_{\Az}} \big)^{2} + \big( \braketq{\psi_{\Aw}} \big)^{2} - 2 \abs{\braket{\psi_{\Az}}{\psi_{\Aw}}}^{2}.
\end{align}
Since $\lambda_{1} \lambda_{2} = \frac{1}{2} \big[ (\lambda_{1} + \lambda_{2})^{2} - (\lambda_{1}^{2} + \lambda_{2}^{2}) \big] = \abs{\braket{\psi_{\Az}}{\psi_{\Aw}}}^{2} - \braketq{\psi_{\Az}} \braketq{\psi_{\Aw}}$, the Cauchy-Schwarz inequality ensures that $\lambda_{1} \lambda_{2} \leq 0$. As $\HHH$ is Hermitian, we have $\norm{\HHH}_{1} = \abs{\lambda_{1}} + \abs{\lambda_{2}}$ and since the eigenvalues have opposite signs we can write it as
\begin{equation}
\label{eq:one-norm}
\norm{\HHH}_{1} = \abs{\lambda_{1} - \lambda_{2}} = \sqrt{(\lambda_{1} - \lambda_{2})^{2}}.
\end{equation}
Expanding the square gives
\begin{align}
(\lambda_{1} - \lambda_{2})^{2} &= \lambda_{1}^{2} + \lambda_{2}^{2} - 2 \lambda_{1} \lambda_{2}\\
&= \big( \braketq{\psi_{\Az}} \big)^{2} + \big( \braketq{\psi_{\Aw}} \big)^{2} - 2 \abs{\braket{\psi_{\Az}}{\psi_{\Aw}}}^{2} + 2\braketq{\psi_{\Az}} \braketq{\psi_{\Aw}} - 2\abs{\braket{\psi_{\Az}}{\psi_{\Aw}}}^{2}\\
&= \big( \braketq{\psi_{\Az}} + \braketq{\psi_{\Aw}} \big)^{2} - 4 \abs{\braket{\psi_{\Az}}{\psi_{\Aw}}}^{2},
\end{align}
which combined with~\eqref{eq:one-norm} implies
\begin{align}
\norm[\Big]{\dya{\psi_{\Az}} - \dya{\psi_{\Aw}} }_{1} &= \sqrt{ \big( \braketq{\psi_{\Az}} + \braketq{\psi_{\Aw}} \big)^{2} - 4 \abs{\braket{\psi_{\Az}}{\psi_{\Aw}}}^{2} }\\
&= \sqrt{ ( \Tr \AAz + \Tr \AAw )^{2} - 4 \norm{\sqrt{\AAz}  \sqrt{\AAw}}_{1}^{2} }.
\end{align}
\end{proof}

Combining Lemmas \ref{lem:max-fidelity} and \ref{lem:trace-norm-fidelity} gives the following result. 
\begin{lemma}
\label{lemMaxPguess223}
For any cq state $\rho_{XB}$ where $X$ is binary,
\begin{equation}
\label{eqnmaxentpguess22}
H_{\textnormal{max}}(X | B) \leq \log \bigg( 1 + \sqrt{1-[2p_{\g}(X|B)-1 ]^2} \bigg).
\end{equation}
\end{lemma}
\begin{proof}
Apply Lemma \ref{lem:trace-norm-fidelity} to the subnormalized states $\sg_0$ and $\sg_1$ appearing in Lemma \ref{lem:max-fidelity} to give
\begin{equation}
\|  \sg_0 - \sg_1  \|_1^2 + 4 \| \sqrt{ \sg_0 } \sqrt{ \sg_1 } \|_1^2 \leq 1.
\end{equation}
Since $\norm{\sqrt{ \sg_0 } \sqrt{ \sg_1 }}_{1} = F(\sg_{0}, \sg_{1})$ combining this inequality with~\eqref{eqnmaxentfid1} and~\eqref{eqnpguesstd1} gives the desired inequality.
\end{proof}

\section{Hybrid of predictive and retrodictive scenarios}\label{scthybrid}

\subsection{Introduction}

In the main text we discussed how to derive WPDRs from preparation and measurement uncertainty relations, which respectively deal with predicting the future and retrodicting the past. In this section we show that our framework can also be applied to scenarios that involves a hybrid (or mixture) of prediction and retrodiction. 

Let us first mention a motivating example from the literature for when this hybrid situation can arise. Ref.~\cite{PhysRevA.85.054101} considered a simple yet insightful scenario involving a MZI where both beamsplitters $\bs_1$ and $\bs_2$ (see, e.g., Fig.~\ref{fgrReg}) may be asymmetric. Since $\bs_1$ is asymmetric, the experimenter has prior knowledge about which path the photon will take. Since $\bs_2$ is asymmetric, the experimenter can use the final measurement outcome of which detector clicked to help retrodict which path the photon took. Ref.~\cite{PhysRevA.85.054101} formulated a WPDR for this scenario, and by the end of this section it will be clear that this falls under our framework. But let us develop the general idea first.

\subsection{General treatment of hybrid scenario}

To treat the hybrid case we will consider our retrodictive framework (discussed in the Methods) and add in the possibility of pre-experiment information about which path the quanton will take. Recall that, in the retrodictive case, we introduced a qubit register $Q'$ that is maximally entangled to the quanton $S$ at time $t_1$, the time just after the quanton enters the interferometer. More precisely $Q'$ is maximally entangled to the interfering subspace $Q$ of $S$. The purpose of $Q'$ is to store a record of the quanton's properties at time $t_1$, so that when the quanton evolves and changes over time, we can still go back to $Q'$ to ask about the quanton's properties at the earlier time. 

The fact that we chose a maximally entangled state is connected to the fact that there is no prior knowledge about the path the quanton will take. But now we are relaxing that assumption, so we will consider a \textit{partially} entangled state. It is useful to think of this partially entangled state as arising from taking the physical state $\rho_{S}^{(1)}$ at time $t_1$, and then applying an isometry $V_c$ that copies the which-path information and stores it in $Q'$. This isometry expands the Hilbert space, mapping $S$ to $Q'S$ as follows:
\begin{equation}
\label{eqnVc23905}
V_c=\sum_{j=0}^1 \ket{j}_{Q'}\ot \dya{j}_S,\quad \rhob^{(1)}_{Q'S} := V_c \rho_{Q}^{(1)} V_c\ad.
\end{equation}
There is a minor technical detail in \eqref{eqnVc23905} that is irrelevant to the MZI but becomes relevant, e.g., in the Franson interferometer. Namely, in \eqref{eqnVc23905}, instead of using the initially-prepared state $\rho_{S}^{(1)}$, which may have support outside of the interfering subspace $Q$, we use the projected and renormalized state $\rho_{Q}^{(1)}$, defined as
\begin{equation}
\label{eqnVc239066}
\rho_{Q}^{(1)} :=   N_1 \cdot (\Pi \rho_{S}^{(1)} \Pi),\quad \text{with }N_1 := 1/\Tr(\Pi \rho_{S}^{(1)}),\quad \text{and }\Pi =\dya{0} +\dya{1} .
\end{equation}
As discussed in the Methods section, the physical motivation behind this projection is that it corresponds to the experimenter post-selecting on the interfering portion of the data. Note that if $\rho_{Q}^{(1)}$ is a pure state from the $XY$ plane of the Bloch sphere (i.e., of the form $(\ket{0}+e^{i\phi}\ket{1})/\sqrt{2}$), then $\rhob^{(1)}_{Q'S}$ is maximally entangled and then we are just back to the retrodictive case discussed in the Methods. The generality in the present treatment comes from the fact that $\rho_{Q}^{(1)}$ is arbitrary.

As in the Methods, we treat the dynamics after time $t_1$ very generally by saying that some \textit{quantum operation} \cite{NieChu00} (completely positive trace non-increasing map) denoted $\AC$ acts on system $S$, mapping $S$ to the joint system $E_1E_2$. The output system does not contain $S$ because the quanton is eventually detected by a detector, at which point we are no longer interested in discussing the quanton's spatial degree of freedom quantum mechanically; we only care where it was detected. The map $\AC$ corresponds to a particular detection event; for concreteness let us say that detector $D_0$ clicking is the associated detection event. The probability for this event is the trace of the state after the action of $\AC$, and upon renormalizing we arrive at the final state
\begin{equation}
\label{eqnrhob2342aab}
\rhob^{D_0}_{Q' E_1 E_2}:=\frac{  (\IC \ot \AC)(\rhob^{(1)}_{Q'S})  }{ \Tr [ (\IC \ot \AC)(\rhob^{(1)}_{Q'S}) ] }.
\end{equation}
To derive WPDRs for the hybrid scenario, we apply the our main uncertainty relation, Eq.~(5) from the main text, to the state $\rhob^{D_0}_{Q' E_1 E_2}$.

Consider the following important special case, where the map $\AC$ involves two sequential steps. First there is interaction between $S$ and an environment $F$ inside the interferometer between times $t_1$ and $t_2$, which corresponds to feeding $S$ through a channel $\EC$, mapping $S$ to $S \En$, obtaining the state $\rhob^{(2)}_{Q'S \En}=(\IC\ot \EC)(\rhob^{(1)}_{Q'S})$. Second, system $S$ is detected at the interferometer output, say at detector $D_0$, which we can model as a map $\BC(\cdot) = \Tr_S [C_0(\cdot)]$, where $C_0$ is the POVM element (acting on $S$ at time $t_2$) associated with detector $D_0$ clicking. Hence we choose $\AC = \BC \circ \EC$, giving the (renormalized) state:
\begin{equation}
\label{eqnrhob2342}
\rhob^{D_0}_{Q' \En}:=\frac{\Tr_S(C_0 \rhob^{(2)}_{Q'S \En})  }{\Tr (C_0 \rhob^{(2)}_{Q'S\En} ) }.
\end{equation}
Now applying our main uncertainty relation to the state $\rhob^{D_0}_{Q' \En}$ and choosing $E_2$ to be trivial and $E_1 =F$ gives
\begin{equation}
\label{eqnJMeur}
H_{\min}(Z_{Q'} | \En )_{\rhob^{D_0}}+\min_{W\in XY}H_{\max}(W_{Q'})_{\rhob^{D_0}}\geq 1,
\end{equation}
where the subscript $\rhob^{D_0}$ means evaluating the entropy on the state in \eqref{eqnrhob2342} and the subscript $Q'$ is used to emphasize that the observables $Z$ and $W$ refer to system $Q'$.

\begin{figure}[t]
\begin{center}
\includegraphics[scale=1.1]{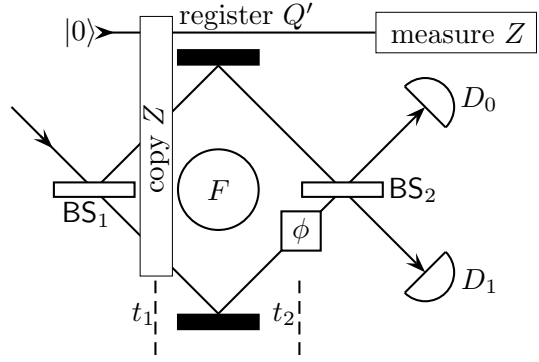}
\caption{%
In the case, e.g., where both $\bs_1$ and $\bs_2$ are asymmetric beam splitters, the distinguishability (denoted $\DC_{Q'}$ in the text) can be measured as follows. The experimenter can copy the which-path information at time $t_1$ to a register $Q'$ (e.g., $Q'$ may be the photon's polarisation) via a controlled-not gate. Then the experimenter can measure the observable on $Q'$ that encodes the which-path information and estimate the probability of correctly guessing this measurement's outcome.\label{fgrReg}}
\end{center}
\end{figure}

At this point we must remark on the physical meaning of a relation such as \eqref{eqnJMeur}. In the special where the state $\rhob^{(1)}_{Q'S}$ was maximally entangled, as considered in the main text, we noted that \eqref{eqnJMeur} can be interpreted as a joint measurement relation. This is because a maximally entangled state is special in that it maps observables on register $Q'$ to the transpose observables on the system of interest. However this interpretation is lost once we relax the form of $\rhob^{(1)}_{Q'S}$, so we can no longer interpret \eqref{eqnJMeur} as a joint measurement relation, in the general case. Rather, one can think of \eqref{eqnJMeur} as a hybrid between a preparation and measurement uncertainty relation.

Regardless, \eqref{eqnJMeur} can be written in the traditional WPDR form. First we note that the isometry $V_c$ in \eqref{eqnVc23905} maps the which-path observable $Z$ of $Q$ onto the corresponding $Z$ observable of $Q'$. So the observable $Z_{Q'}$ is effectively the which-path observable, and we can define path distinguishability by 
\begin{equation}
\label{eqnDCQp209}
\DC_{Q'}:= 2 p_{\g}(Z_{Q'}|F)_{\rhob^{D_0}}-1, \quad \text{or in other words,   } H_{\min}(Z_{Q'} | \En )_{\rhob^{D_0}} = 1- \log (1+\DC_{Q'} ).
\end{equation}
Note that $\DC_{Q'}$ is a generalization of the input distinguishability $\DC_i$ defined in the main text. Experimentally measuring $\DC_{Q'}$ involves the procedure shown in Fig.~\ref{fgrReg} for the special case of the MZI. That is, the experimenter perfectly copies the which-path information to a register and then measures the register. Note that a register is necessary in this case since other procedures to measure distinguishability discussed in the main text - such as removing the second beam splitter or randomly inserting a blocker in one arm - would not capture the inherent which-path asymmetry associated with both beam splitters.

Moving onto the which-phase observable $W$, we prove a general and powerful lemma that the second term in \eqref{eqnJMeur} corresponds precisely to the fringe visibility $\VC$. Let us remind the reader that the latter is given by 
\begin{equation}
\label{eqnfringvis39402}
\VC  = \frac{p^{D_0}_{\max}-p^{D_0}_{\min}}{p^{D_0}_{\max} + p^{D_0}_{\min}}
\end{equation}
where $p^{D_0}$ is the probability for the quanton to be detected at detector $D_0$, $p^{D_0}_{\max} = \max_{\phi} p^{D_0}$, and $p^{D_0}_{\min} = \min_{\phi} p^{D_0}$. To establish this connection, we need to restrict the form of the channel $\EC_S = \Tr_F \circ \EC$ acting on the quanton such that $\EC_S$ is path preserving. We emphasize that the following lemma is a generalization of Lemma 2 from the main text and hence its proof implies Lemma 2. See the next subsection for the proof.
\begin{lemma}
\label{lem22}
Consider any binary interferometer. Let $\EC_S = \Tr_F \circ \EC$ denote the quantum channel describing the quanton's interaction with the environment $F$ inside the interferometer between times $t_1$ and $t_2$, and let $\GC(\cdot) = \Pi (\cdot)\Pi$ denote the map that projects onto the subspace $\Pi =\dya{0}+\dya{1}$. Suppose that $\EC_S$ is path-preserving, i.e., $\EC_S(\dya{0})=\dya{0}$ and $\EC_S(\dya{1})=\dya{1}$ and furthermore suppose that $\EC_S$ commutes with $\GC$, i.e., $\GC \circ \EC_S = \EC_S \circ \GC$. Then
\begin{equation}
\label{eqnVisMaxRelationMeth2}
\min_{W\in XY} H_{\max}(W_{Q'})_{\rhob^{D_0}} = \log (1+\sqrt{1-\VC^2}),
\end{equation}
where $H_{\max}(W_{Q'})$ is evaluated for the state $\rhob_{Q'}^{D_0}=\Tr_{\En}(\rhob_{Q' \En}^{D_0})$ from \eqref{eqnrhob2342}.
\end{lemma}

Now making the assumption in Lemma~\ref{lem22} regarding the form of the interaction $\EC_S$ allows us to rewrite \eqref{eqnJMeur}, using \eqref{eqnDCQp209} and \eqref{eqnVisMaxRelationMeth2}, as follows:
\begin{equation}
\label{eqnWPDRdcqp34}
\DC_{Q'}^2+\VC^2\leq 1.
\end{equation}
This WPDR holds for generic binary interferometers.

Ref.~\cite{PhysRevA.85.054101} derived a WPDR similar to \eqref{eqnWPDRdcqp34} but for the special case of a MZI where both $\bs_1$ and $\bs_2$ are possibly asymmetric. They restricted to interactions with the environment $F$ that had the form of a controlled-unitary, which are path-preserving and hence are included in the class of dynamics that we assumed to derive \eqref{eqnWPDRdcqp34}. Of course \eqref{eqnWPDRdcqp34} applies very generally to binary interferometers, but it can be applied to the MZI (note that $Q = S$ in the MZI case) where both $\bs_1$ and $\bs_2$ are asymmetric. Thus, we find that the result in Ref.~\cite{PhysRevA.85.054101} can be understood as an entropic uncertainty relation, namely a special case of \eqref{eqnJMeur} [from which we derived \eqref{eqnWPDRdcqp34}]. To be more precise, Ref.~\cite{PhysRevA.85.054101} also generalized their relation to allow for non-optimal strategies for measuring $F$; we do not treat this generalization here. We believe the most important conceptual advance of Ref.~\cite{PhysRevA.85.054101} was to prove a WPDR that applies to a scenario that - in our language - is a hybrid of preparation and measurement uncertainty. What we have emphasized in this section is that our framework naturally extends to this hybrid scenario.

\subsection{Proof of Lemma~\ref{lem22} (generalized version of Lemma 2)}\label{sctProofLem2}

To prove Lemma~\ref{lem22} we will make use of the following lemma.

\begin{lemma}
\label{lem22231}
Let $\QC$ be a qubit quantum channel, i.e., whose input and output are operators on a 2-dimensional Hilbert space, and suppose that $\QC(\dya{j}) = \dya{j}$ for $j=0,1$. Likewise let $\RC_{\phi}$ be a qubit quantum channel whose action is given by
\begin{equation}
\label{eqnRphidef1}
\RC_{\phi}(\cdot) = U_{\phi}(\cdot)U_{\phi}\ad,\quad\text{with  }U_{\phi} = \dya{0}+e^{i\phi}\dya{1}.
\end{equation}
Then, for any $\phi$, $\QC$ and $\RC_{\phi}$ commute, i.e., $\RC_{\phi} \circ \QC = \QC \circ \RC_{\phi}$.
\end{lemma}
\begin{proof}
Since $\QC$ is unital and furthermore preserves the $Z$-basis, its action on the Bloch sphere can only involve a rotation $\RC_{\theta}$ about the $Z$-axis composed with a shrinking $\SC$ of the Bloch sphere, and this shrinking must be cylindrically symmetric about the $Z$-axis (see, e.g., Ref.~\cite{King01}). The rotation $\RC_{\theta}$ obviously commutes with the rotation $\RC_{\phi}$, and likewise $\RC_{\phi}$ commutes with the shrinking $\SC$ due to the cylindrically symmetry of $\SC$.
\end{proof}

Now we prove Lemma~\ref{lem22}, which relates the fringe visibility to the max-entropy of the which-phase observable in our ``hybrid'' framework. Lemma~\ref{lem22} generalizes Lemma 2 from the main text, which is the corresponding result for our retrodictive framework.

\begin{proof}
In the formula for $\VC$ in \eqref{eqnfringvis39402}, the notation
\begin{equation}
\label{eqnpd0def}
p^{D_0} = \Tr \big[ C_0 \RC_{\phi} (\rho^{(2)}_Q )\big]
\end{equation}
refers to the probability for detector $D_0$ to click when the quanton's state at time $t_2$ (the time just before the phase-shift $\phi$ is applied) is $\rho^{(2)}_Q$. This state was defined in Eq.~(25) of the Methods section as
\begin{equation}
\label{eqnrhoq2def33}
\rho^{(2)}_Q = N_2 \cdot ( \Pi \rho_S^{(2) } \Pi ),\quad \text{with }N_2 : = 1/\Tr (\Pi \rho_S^{(2) }  ),
\end{equation}
which notes that the experimenter post-selects on the interfering subspace, associated with projector $\Pi$. Defining $\tilde{C}_0 := \Pi C_0 \Pi$ and taking the maximum of \eqref{eqnpd0def} over $\phi$ gives
\begin{align}
\label{eqnpd0max1235}
p^{D_0}_{\max} &= N_2 \max_{\phi} \Tr  \big[  \Ct_0 \RC_{\phi} ( \Pi \rho^{(2)}_S \Pi )  \big] \\
&= N_2 \max_{\phi} \Tr   \big[  \Ct_0 \RC_{\phi} ( \Pi \EC_S (\rho^{(1)}_S) \Pi ) \big] \\
&= N_2 \max_{\phi} \Tr   \big[  \Ct_0 \RC_{\phi} (\EC_S ( \Pi \rho^{(1)}_S \Pi)  ) \big] \\
&= (N_2/N_1) \max_{\phi} \Tr  \big[  \Ct_0 \RC_{\phi} ( \EC_S (  \rho^{(1)}_Q )  ) \big] \\
&= (N_2/N_1) f \big( \widehat{\phi}\hspace{1pt} \big),\quad\text{   where }f(\phi):=\Tr  \big[  \Ct_0 \RC_{\phi} ( \EC_S (  \rho^{(1)}_Q )  ) \big], 
\end{align}
and we use $\widehat{\phi}$ to denote that phase that maximises $f(\phi)$. Now, by thinking of $f(\phi)$ as the inner product between two vectors in the Bloch sphere, one can see that the phase $\phi$ that minimizes $f(\phi)$ is 180 degrees added to the phase that maximizes it. So we have
\begin{align}
\label{eqnpd0max1235}
p^{D_0}_{\min} &= (N_2/N_1) \min_{\phi} f(\phi) = (N_2/N_1) f \big(\widehat{\phi} +\pi \big).
\end{align}
Hence from \eqref{eqnfringvis39402} we compute the fringe visibility to be
\begin{equation}
\label{eqnVformula3902}
\VC = \frac{  f \big( \widehat{\phi}\hspace{1pt} \big) -f \big(\widehat{\phi} +\pi \big)}{  f \big( \widehat{\phi}\hspace{1pt} \big)+f \big(\widehat{\phi} +\pi \big)}.
\end{equation}

Now consider the left-hand side of \eqref{eqnVisMaxRelationMeth2}, which we write as
\begin{align}
\label{eqnHmax3219}
\min_{W\in XY} H_{\max}(W_{Q'})_{\rhob^{D_0}}= \log \Big(1+\sqrt{1-\VC_{Q'}^2}\Big),
\end{align}
which defines the visibility-like quantity $\VC_{Q'}$, and ultimately we wish to show that $\VC_{Q'} =\VC$. The formula for the unconditional max-entropy is $H_{\max}(\{p_j\}) = 2\log \sum_j \sqrt{p_j}$, which implies that
\begin{align}
\label{eqnHmax234095}
H_{\max}(W_{Q'})_{\rhob^{D_0}} = \log \bigg(1+\sqrt{1-(2 \Pr(w_+)_{\rhob^{D_0}}-1)^2} \bigg) 
\end{align}
where $\Pr(w_+)_{\rhob^{D_0}} := \mte{w_+}{\rhob^{D_0}_{Q'}}$. Comparing \eqref{eqnHmax3219} with \eqref{eqnHmax234095}, we see that
\begin{align}
\label{eqnVq239}\VC_{Q'} &= 2 \max_{W\in XY }\Pr(w_+)_{\rhob^{D_0}}-1.
\end{align}
Using the formula for the state $\rhob^{D_0}_{Q'}$ in \eqref{eqnrhob2342}, we have
\begin{align}
\label{eqnpmax2333aaaa}\Pr(w_+)_{\rhob^{D_0}} 
 =   \frac{\mte{w_+}{\Tr_S(C_0 \rhob^{(2)}_{Q'S }) } }{\Tr (C_0 \rhob^{(2)}_{S} )} 
=  \frac{\Tr_{Q'S} [(\dya{w_{+}}\ot C_0) \cdot  (\IC\ot \EC_S ) ( V_c \rho_{Q}^{(1)}V_c\ad )]  }{ \Tr (C_0 \rhob^{(2)}_{S} )}.
\end{align}
Now let $\ket{w_{+}} = (\ket{0}+e^{i\phi }\ket{1})/\sqrt{2} = U_{\phi } \ket{+}$, and let us maximise \eqref{eqnpmax2333aaaa} over all $W$ in the $XY$ plane, which corresponds to maximising over $\phi$. Noting that the denominator on the right-hand side of \eqref{eqnpmax2333aaaa} is independent of $W$, we have
\begin{align}
\label{eqnpmax2333aeee}\max_{W\in XY}\Pr(w_+)_{\rhob^{D_0}} &= \frac{1}{\Tr (C_0 \rhob^{(2)}_{S} )}\max_{W\in XY} \Tr_{S} [C_0 \cdot \EC_S ( \Tr_{Q'} (  ( \dya{w_{+}} \ot \id_S)  V_c \rho_{Q}^{(1)}V_c\ad ))] \\
\label{eqnpmax2333abc}&=\frac{1}{\Tr (C_0 \rhob^{(2)}_{S} )} \max_{\phi} \Tr_{S} [C_0 \cdot \EC_S ( \Tr_{Q'} (  ( U_{\phi}\dya{+} U_{\phi}\ad \ot \id_S)  V_c \rho_{Q}^{(1)}V_c\ad ))] \\ 
\label{eqnpmax2333abcd}&=\frac{1}{2\Tr (C_0 \rhob^{(2)}_{S} )} \max_{\phi} \Tr_S [C_0 \cdot \EC_S ( U_{\phi}\ad \rho_{Q}^{(1)} U_{\phi}) ] \\
\label{eqnpmax2333abcde}&=\frac{1}{2\Tr (C_0 \rhob^{(2)}_{S} )} \max_{\phi} \Tr_S [ C_0 \cdot \RC_{\phi} (\EC_S (  \rho_{Q}^{(1)}  ) ) ] \\
\label{eqnpmax2333abcdef}&=\frac{1}{2\Tr (C_0 \rhob^{(2)}_{S} )}  f \big( \widehat{\phi}\hspace{1pt} \big),
\end{align}
where \eqref{eqnpmax2333abcde} invoked Lemma~\ref{lem22231}. Next, using the fact that $ \rhob^{(1)}_{S} = \RC_{\phi} (\rhob^{(1)}_{S})$ is diagonal in the standard basis, and furthermore that $\rho^{(1)}_{Q}+ \RC_{\pi} (\rho^{(1)}_{Q}) = 2\rhob^{(1)}_{S}$, we have
\begin{align}
\label{eqnpmax2333c231}
2\Tr (C_0 \rhob^{(2)}_{S}) =  2\Tr (\Ct_0 \EC_S( \rhob^{(1)}_{S}) ) = 2\Tr \big[ \Ct_0 \EC_S( \RC_{\widehat{\phi}} (\rhob^{(1)}_{S}) )\big] = \Tr \big[ \Ct_0 \EC_S( \RC_{\widehat{\phi}}(\rho^{(1)}_{Q})+ \RC_{\widehat{\phi}+\pi } (\rho^{(1)}_{Q}) ) \big]=  f \big( \widehat{\phi}\hspace{1pt} \big) +f \big(\widehat{\phi} +\pi \big) .
\end{align}
Combining \eqref{eqnVq239}, \eqref{eqnpmax2333abcdef}, and \eqref{eqnpmax2333c231} gives
\begin{align}
\label{eqnVq232219}\VC_{Q'} &=  \frac{  2  f \big( \widehat{\phi}\hspace{1pt} \big) }{  f \big( \widehat{\phi}\hspace{1pt} \big) +f \big(\widehat{\phi} +\pi \big) }-1,
\end{align}
which is equivalent to the formula in \eqref{eqnVformula3902}, and hence completes the proof.\end{proof}

\section{Enhanced visibility}\label{sct4A}

In this section, we consider two examples from the literature, Ref.~\cite{Englert2000337} and \cite{Banaszek:2013fk}, of WPDRs where the visibility is enhanced by utilizing a portion of the environment. This corresponds to system $E_2$ being non-trivial in our main WPDR, Eq.~(5) from the main text. We show how both literature results fit into our framework.

\subsection{Quantum erasure}\label{sct4Ab}


\subsubsection{Results of Ref.~\cite{Englert2000337}}

Ref.~\cite{Englert2000337} derived some WPDRs in which the visibility is enhanced by conditioning on a measurement on the environment. This scenario is called quantum erasure since it aims to erase the which-path information stored in the environment. Here we show that this scenario can be treated in our framework, and hence, that the main results of \cite{Englert2000337} can be viewed as entropic uncertainty relations.


Ref.~\cite{Englert2000337} considered interferometers that can be modeled as a qubit, which are slightly less general than our notion of binary interferometers, where the interfering subspace is a qubit living inside a larger space. While it should be clear that the treatment can be extended to binary interferometers, for simplicity we will present the treatment as in Ref.~\cite{Englert2000337}, as follows.

Suppose the qubit system of interest $Q$ is initially in state $\rho_{Q}^{(1)}$ at time $t_1$ (see, e.g., Fig.~1 from the main text). Ref.~\cite{Englert2000337} allowed the system $Q$ to interact with an environment $\En$ resulting in a bipartite state $\rho_{Q \En}^{(2)} = \EC_{\text{int}}(\rho_{Q}^{(1)})$ at time $t_2$, and then an observable $\Gamma $ on system $\En$ is measured. We can represent $\Gamma$ as a set of orthogonal projectors $\{\Gamma_k\}$ with $\sum_k \Gamma_k = \id_{\En}$. (We do not lose generality by assuming the $\Gamma_k$ are projectors instead of arbitrary positive operators, since system $\En$ is arbitrary and any POVM can be thought of as a projective measurement on an enlarged Hilbert space.) Obtaining outcome $k$ leaves system $Q$ in the conditional state 
\begin{align}
\rho_{Q,k}^{(2)} = \frac{1}{g_k}\Tr_{\En} \big[ (\id_Q \ot \Gamma_k) \rho_{Q \En}^{(2)} \big], \quad \text{with }g_k :=\Tr \big[ (\id_Q \ot \Gamma_k)  \rho_{Q \En}^{(2)} \big] .
\end{align}
One can define the path predictability and fringe visibility associated with this conditional state as
\begin{align}
\PC_k &:= 2 p_{\g}(Z)_{k} - 1\\
\VC_k & := \frac{p^{D_0}_{\max,k} - p^{D_0}_{\min,k}}{p^{D_0}_{\max,k} + p^{D_0}_{\min,k}}
\end{align}
where the subscript $k$ just means evaluating the quantity for the state $\rho_{Q,k}^{(2)}$. Ref.~\cite{Englert2000337} now defined the average predictability and visibility (i.e., averaged over all measurement outcomes) as
\begin{align}
\PC(\Gamma) := \sum_k g_k \PC_k,\quad \VC(\Gamma)  := \sum_k g_k \VC_k.
\end{align}
Ref.~\cite{Englert2000337} noted that maximizing $\PC(\Gamma)$ over all $\Gamma$ gives the distinguishability, while they defined a quantity called ``coherence'' as the supremum over all $\Gamma$ of $\VC(\Gamma)$, as follows
\begin{align}
\DC =\max_{\Gamma} \PC(\Gamma), \quad \CC := \sup_{\Gamma} \VC(\Gamma). 
\end{align}
They noted the hierarchies $\PC\leq \PC(\Gamma) \leq \DC$ and $\VC\leq \VC(\Gamma) \leq \CC$. The two main results that were highlighted in \cite{Englert2000337} were the WPDRs
\begin{align}
\label{eqnPVmain44}\PC(\Gamma)^2+\VC(\Gamma)^2 &\leq 1,\\
\label{eqnPVmain45}\PC^2+\CC^2 &\leq 1,
\end{align}
where \eqref{eqnPVmain44} holds for any choice of $\Gamma$. Actually, Ref.~\cite{Englert2000337} noted that \eqref{eqnPVmain44} implies \eqref{eqnPVmain45} by taking the supremum such that $\VC(\Gamma)$ approaches $\CC$. So let us focus on proving \eqref{eqnPVmain44}.

\subsubsection{Our treatment}

The overall dynamics described above can be separated into three steps:
\begin{enumerate}
  \item The system $Q$ interacts with an environment $\En$, via CPTP map $\EC_{\text{int}}$.
  \item System $\En$ is measured and the outcome is stored in a register $R$, via CPTP map $\EC_{\text{meas}}$.
  \item The experimenter uses this measurement result to enhance the visibility on system $Q$ (i.e., to sort the data point into a sub-ensemble and determine the optimal phase shift for that sub-ensemble). This is modelled as a CPTP map $\EC_{\enh}$ that couples $R$ to $Q$.
\end{enumerate}
The overall CPTP map $\EC$ is a composition of these three maps:
\begin{equation}
\EC = \EC_{\enh} \circ \EC_{\text{meas}} \circ \EC_{\text{int}}. 
\end{equation}
As noted above, the interaction with $\En$ results in the state $\rho_{Q\En}^{(2)} := \EC_{\text{int}}(\rho_{Q}^{(1)})$. Next, $\EC_{\text{meas}}$ performs the projective measurement $\Gamma=\{\Gamma_k\}$ on system $\En$ and stores the outcome in two (redundant) registers $R$ and $R'$:
\begin{equation}
\rho_{QRR'}^{(3)}:= \EC_{\text{meas}}(\rho_{Q\En}^{(2)}) = \sum_k g_k \rho_{Q,k}^{(2)} \ot \dya{k}_R \ot \dya{k}_{R'}
\end{equation}
where state $\ket{k}$ corresponds to obtaining outcome $k$ from measuring $\Gamma$, and the set $\{\ket{k} \}$ forms an orthonormal basis on the register Hilbert space. The point of having two registers is that one register will act as system $E_1$ from the main text - to be used to enhance the distinguishability - while the other will act as system $E_2$ from the main text - to be used to enhance the visibility.

For each measurement outcome $k$, we wish to obtain the full visibility that is available, so we allow the experimenter to choose the optimal basis $W_k$ in the $XY$ plane of the Bloch sphere for each $k$. We can think of this as allowing the experimenter, given the outcome $k$, to rotate the system $Q$ via a unitary $U^Z_k$ that is diagonal in the $Z$ basis. Suppose this unitary is tailored to rotate the optimal basis $W_k$ to the $X$ basis, i.e., $X =   U^Z_k W_k (U^Z_k)\ad$ for each $k$. Accounting for all possible values of $k$, the overall unitary is a controlled unitary $U_{\enh} := \sum_k U^Z_k \ot \dya{k}_{R}$ where $R$ acts as the control system. Hence the action of the map that enhances the visibility is:
\begin{equation}
\rho_{QRR'}^{(4)}:= \EC_{\enh}(\rho_{QRR'}^{(3)}) = U_{\enh} \rho_{QRR'}^{(3)} U_{\enh}\ad = \sum_k  g_k \rhot_{Q,k}^{(2)}  \ot \dya{k}_R \ot \dya{k}_{R'},\quad \text{with }\rhot_{Q,k}^{(2)}:= U^Z_k  \rho_{Q,k}^{(2)} (U^Z_k)\ad .
\end{equation}
Finally, we apply our main WPDR, Eq.~(5) from the main text, to the state $\rho_{QRR'}^{(4)}$. Specifically we choose $E_1=R$ and $E_2 =R' =R$ noting that $R'$ and $R$ are identical copies, giving
\begin{equation}
\label{eqnQuEr1}
H_{\min}(Z|R)_{\rho^{(4)}}+H_{\max}(X|R)_{\rho^{(4)}} \geq 1,
\end{equation}
where the subscript $\rho^{(4)}$ emphasizes that the entropy terms are evaluated for the state $\rho_{QRR'}^{(4)}$, for which $X$ is the basis that achieves the minimization in $\min_{W\in XY}H_{\max}(W|R)$.

We now show how our quantum erasure relation in \eqref{eqnQuEr1} implies \eqref{eqnPVmain44}, which in turn implies \eqref{eqnPVmain45} as noted previously. First note that
$$\PC(\Gamma) = \sum_k g_k\PC_k = \sum_k g_k [2 p_{\g}(Z)_{k} - 1] =2 \bigg[ \sum_k g_k  p_{\g}(Z)_{k} \bigg] - 1 = 2 p_{\g}(Z|R)_{\rho^{(4)}} - 1,$$ 
which implies that
\begin{equation}
\label{eqn380363a}
H_{\min}(Z|R)_{\rho^{(4)}} = - \log p_{\g}(Z|R)_{\rho^{(4)}} = 1 - \log [1+\PC(\Gamma)].
\end{equation}
Next, using the relation \eqref{eqnmaxentpguess22} and noting that $\VC(\Gamma) = 2p_{\g}(X|R)_{\rho^{(4)}}-1 $, we have that
\begin{equation}
\label{eqn380363ab}
H_{\max}(X|R)_{\rho^{(4)}} \leq \log(1+\sqrt{1-\VC(\Gamma)^2}).
\end{equation}
Substituting \eqref{eqn380363a} and \eqref{eqn380363ab} into \eqref{eqnQuEr1} gives \eqref{eqnPVmain44}.

We remark that the uncertainty relation \eqref{eqnQuEr1} corresponds to the ``preparation uncertainty'' scenario discussed in the main text, addressing the question of the predictability of the measurement at the interferometer \textit{output}. Hence the results of Ref.~\cite{Englert2000337} are of the preparation uncertainty variety.

\subsection{Polarization-enhanced visibility and discussion of Ref.~\cite{Banaszek:2013fk}}\label{sct4Ab}

\subsubsection{Result of Ref.~\cite{Banaszek:2013fk}}

The aim of this section is to show that the main result of Ref.~\cite{Banaszek:2013fk} can be viewed as, or is a direct consequence of, the uncertainty relation for the min- and max-entropies, and hence is covered by our framework. Their result is a WPDR for a MZI where a fairly general interaction occurs inside the interferometer between the photon's spatial degree of freedom $Q$, its polarization $P$, and an environment $\En$. 

In their WPDR they replaced the usual fringe visibility with a quantity they called generalized visibility with the motivation that it provides a stronger bound on the path distinguishability (via the WPDR). We will denote their generalized visibility as $\VC_B$, where the subscript $B$ refers to the first author of Banaszek et al.~\cite{Banaszek:2013fk}. (Similarly we will use $\DC_B$ for their path distinguishability.) Their generalized visibility $\VC_B$ was written in an abstract form; however, its operational meaning was not stated. Here we give an operational meaning to $\VC_B$, showing that it is directly proportional to the guessing probability for the which-phase observable, given the optimal measurement on the photon's polarisation $P$ (see below for the precise statement). Hence the visibility is \textit{enhanced} by gaining further information from measuring the polarisation.

We remark that enhancing the visibility in this way can be done either in our predictive or retrodictive framework. In the previous section we discussed Ref.~\cite{Englert2000337}, which considered the predictive scenario. On the other hand, it turns out that Ref.~\cite{Banaszek:2013fk} considered the retrodictive (or joint measurement) scenario, where one uses the final measurement outcomes to retrodict the photon's properties at an earlier time. Thus, taken together, Refs.~\cite{Englert2000337} and \cite{Banaszek:2013fk} nicely illustrate how visibility can be enhanced in the two different scenarios.

Ref.~\cite{Banaszek:2013fk} allowed for an interaction to occur within the interferometer coupling $Q$ to $P$ and $\En$. We can model this as an isometry $V$ that maps states on $Q$ at time $t_1$ to states on the larger system $Q\En PT$ at time $t_2$, where $T$ acts as a purifying system for the overall state. Ref.~\cite{Banaszek:2013fk} assumed that $V$ is path-preserving. From $V$ we may define two complementary quantum channels
\begin{equation}
\label{eqnCondStatesF}
\FC(\cdot) = \Tr_{Q PT }[V(\cdot)V\ad ], \quad \FC_c(\cdot) = \Tr_{F }[V(\cdot)V\ad ],
\end{equation}
which respectively map $Q$ at time $t_1$ to $F$ (or $QPT$) at time $t_2$. Now, consider the (normalized) conditional states on the environment at time $t_2$ respectively associated with the $\ket{0}$ and $\ket{1}$ input states on $Q$:
\begin{equation}
\label{eqnCondStatesF}
\rho_{\En}^0 = \FC(\dya{0}),\quad \rho_{\En}^1 = \FC (\dya{1}).
\end{equation}
The main result of Ref.~\cite{Banaszek:2013fk} is a WPDR of the form
\begin{equation}
\label{eqn194903Ban}
\DC_B^2+\VC_B^2 \leq 1,
\end{equation}
where 
\begin{align}
\label{eqn194903BanDi}
\DC_B &:= \frac{1}{2} \norm[\Big] {\rho_{\En}^0  - \rho_{\En}^1 }_{1} ,\\
\label{eqnVgAbstract1}
\VC_B &:= \norm[\Big]{\sqrt{\rho_{\En}^0} \sqrt{\rho_{\En}^1}}_{1}.
\end{align}
To clarify, Ref.~\cite{Banaszek:2013fk} wrote $\VC_B$ in several ways, one of which was $\VC_B = \max_{\ket{\psi^{0}},\ket{\psi^{1}}} | \ip{\psi^{0}}{\psi^{1}}|$ where the maximisation is over all purifications $\ket{\psi^{0}}$ and $\ket{\psi^{1}}$ of $\rho_{\En}^0$ and $\rho_{\En}^1$ respectively. But this form of $\VC_B$ is equivalent to that in \eqref{eqnVgAbstract1} due to Uhlmann's theorem (see, e.g., \cite{NieChu00}).

\subsubsection{Our treatment}

We will reinterpret \eqref{eqn194903Ban} within the context of the our complementary guessing game discussed in the main text. In particular let us consider the retrodictive scenario, discussed in the main text, where a sender Bob inputs states randomly (with equal probability for each state in the ensemble) into the interferometer chosen from one of two ensembles, $Z_i = \{\ket{0}, \ket{1}\}$ or $W_i = \{\ket{w_{\pm}}\}$. Alice tries to guess which state Bob sent (given that Bob announces which ensemble he is considering, $Z_i$ or $W_i$). With this scenario in mind, we can rewrite the path distinguishability as
\begin{align}
\label{eqn194903BanDalt}
\DC_B = 2 p_{\g}(Z_i | \En)_{\FC}-1
\end{align}
where the subscript $\FC$ refers to the relevant quantum channel. Mathematically speaking,~\eqref{eqn194903BanDalt} follows from \eqref{eqnpguesstd1}.

Similarly we give an operational interpretation for $\VC_B$ with the following lemma.
\begin{lemma}
\label{lemVGopmeaning}
Let the interaction $V$ be path-preserving, then
\begin{equation}
\label{eqnVg324903}
\VC_B = \max_{W\in XY} [2p_{\g}(W_i | QPT)_{\FC_c}-1]
\end{equation}
where $p_{\g}(W_i | QPT)_{\FC_c}$ denotes the probability of guessing $W_i=\{\ket{w_{\pm}} \}$ correctly given the optimal measurement on the joint system $QPT$, i.e., the output of the channel $\FC_c$.
\end{lemma}
See the next subsection for the proof. This says that $\VC_B$ can be interpreted like an \textit{input} visibility, measuring how well Alice can distinguish between the different $W_i$ states that Bob sends, by doing a measurement on the output system $QPT$. In this sense it quantifies how much phase information or wave-like information gets transmitted from the input $Q$ to the output $QPT$. This clear operational interpretation sheds light on the meaning on $\VC_B$, which was previously only given in an abstract form in \cite{Banaszek:2013fk}.

Now we show how \eqref{eqn194903Ban} can be recast in terms of the min- and max-entropic uncertainty relation. In the retrodictive scenario (as discussed in the Methods) we consider the Choi-Jamio\l{}kowski state obtained from applying $V$ to half of a maximally-entangled state $\ket{\Phi}$ on $Q'Q$,
\begin{equation}
\label{eqn194903sfioi3}
\ket{\Lm}_{Q'Q\En PT} = (\id \ot V) \ket{\Phi}.
\end{equation}
Now we apply the uncertainty relation for the min- and max-entropies to the state $\ket{\Lm}_{Q'Q\En PT}$, giving
\begin{equation}
\label{eqn194rewrwt53aad}
H_{\min}(Z_{Q'} | \En)_{\ket{\Lm}} + \min_{W\in XY} H_{\max}(W_{Q'} |QPT)_{\ket{\Lm}}\geq 1.
\end{equation}
As discussed in the Methods section, the observables $Z_{Q'}$ and $W_{Q'}$ on $Q'$ get mapped to the input ensembles $Z_i$ and $W_i$ on $Q$ (more precisely, one takes the transpose, but this does not matter since we are minimizing $W$ over the $XY$ plane). Hence \eqref{eqn194rewrwt53aad} can be rewritten as
\begin{equation}
\label{eqn194rewrwt53}
H_{\min}(Z_i | \En)_{\FC} + \min_{W\in XY} H_{\max}(W_i | Q PT)_{\FC_c}\geq 1,
\end{equation}
where it should be clear that \eqref{eqn194rewrwt53} refers to the systems $F$ and $QPT$ \textit{after} the interaction $V$, which we emphasise with the subscripts $\FC$ and $\FC_c$ indicating the relevant quantum channels that map the input to the output. Note that this uncertainty relation is a special case of our main result given in the main manuscript, which can be seen by choosing:
\begin{equation}
\label{eqnE1E2ban}
E_1=F,\quad E_2 = QPT.
\end{equation}
Also, we remind the reader that \eqref{eqn194rewrwt53} can be interpreted as a joint measurement uncertainty relation in the following sense. The input ensembles $Z_i$ and $W_i$ are used to calibrate the measurement apparatus, i.e., to assess how well it can measure $Z$ and $W$. Equation~\eqref{eqn194rewrwt53} says the output of the apparatus, $Q\En PT$, cannot provide full information about both the $Z$ and $W$ observables on the input. 

To show that \eqref{eqn194rewrwt53} implies \eqref{eqn194903Ban} we use the relations
\begin{align}
\label{eqn19435631}
H_{\min}(Z_i | \En)_{\FC} &= 1- \log (1+\DC_B),\\
\label{eqn194356315}
\min_{W\in XY} H_{\max}(W_i | Q PT)_{\FC_c}&\leq \log (1+\sqrt{1-\VC_B^2}),
\end{align}
where \eqref{eqn19435631} follows from the relation between the min-entropy and the guessing probability, and \eqref{eqn194356315} follows from combining Lemmas~\ref{lemMaxPguess223} and~\ref{lemVGopmeaning}. Inserting these into \eqref{eqn194rewrwt53} gives \eqref{eqn194903Ban}. This shows that the main result of \cite{Banaszek:2013fk} is an entropic uncertainty relation in disguise, following directly from~\eqref{eqn194rewrwt53}.

\subsubsection{Proof of Lemma~\ref{lemVGopmeaning}}

We first note that the formula in \eqref{eqnVg324903} can be rewritten as
\begin{equation}
\label{eqnVg32490321}
\VC_B = \max_{W\in XY} [2p_{\g}(W_{Q'} | QPT)_{\ket{\Lambda}}-1]
\end{equation}
where $p_{\g}(W_{Q'} | QPT)_{\ket{\Lambda}}$ denotes the probability of correctly guessing the observable $W$ on system $Q'$ given the optimal measurement on $QPT$, for the state $\ket{\Lambda}$ in \eqref{eqn194903sfioi3}. Now to prove \eqref{eqnVg32490321}, we apply~\eqref{eqn194906} in the following lemma, noting that the state $\ket{\Lm}$ in \eqref{eqn194903sfioi3} has the special form assumed in this lemma since $V$ is path-preserving.
\begin{lemma}
\label{lem2362}
Let $\ket{\Psi}_{ABC}$ be a tripartite pure state such that the reduced state on $A$ and $B$ is classical-quantum: $\rho_{AB} = \ketbraq{0} \otimes \sigma_{0} + \ketbraq{1} \otimes \sigma_{1}$, where $\sigma_{0}$ and $\sigma_{1}$ are subnormalized such that $\Tr \sigma_{0} + \Tr \sigma_{1} = 1$. Let $W$ correspond to the outcome of a projective measurement in the $XY$ plane performed on $A$. Then it holds that
\begin{equation}
\label{eqn194906}
2 \norm{\sqrt{\sigma_{0}} \sqrt{ \sigma_{1} }}_1 =2p_{\g}(W | C)_{\ket{\psi}} -1.
\end{equation}
\end{lemma}
\begin{proof}
Since the quantity $p_{\g}(W | C)$ is invariant under unitaries on $C$ we can choose $\ket{\Psi}_{ABC}$ to be an arbitrary purification of $\rho_{AB}$. We find it convenient to split up $C$ into $C_{1}$ and $C_{2}$ which purify $B$ and $A$, respectively:
\begin{equation}
\ket{\Psi}_{ABC_{1} C_{2}} = \ket{0}_{A} \ket{\psi_{\sigma_{0}}}_{BC_{1}} \ket{0}_{C_{2}}  + \ket{1}_{A} \ket{\psi_{\sigma_{1}}}_{BC_{1}} \ket{1}_{C_{2}} .
\end{equation}
Without loss of generality the two orthonormal basis states associated with $W$ can be written as
\begin{align}
\ket{w_{\pm}} = \frac{1}{\sqrt{2}} \Big( \ket{0} \pm e^{i \theta} \ket{1} \Big),
\end{align}
where $\theta \in [0, 2 \pi]$ is a parameter. Performing the measurement of $W$ leads to the following cq state
\begin{align}
\rho_{W B C_{1} C_{2}} &= \frac{1}{2} \ketbraq{w_+}_{W} \otimes \Big( \ketbraq{\psi_{\sigma_{0}}}_{B C_{1}} \otimes \ketbraq{0}_{C_{2}} + e^{i \theta} \ketbra{\psi_{\sigma_{0}}}{\psi_{\sigma_{1}}}_{BC_{1}} \otimes \ketbra{0}{1}_{C_{2}}\notag\\
&\hspace{7pt}+ e^{- i \theta} \ketbra{\psi_{\sigma_{1}}}{\psi_{\sigma_{0}}}_{BC_{1}} \otimes \ketbra{1}{0}_{C_{2}} + \ketbraq{\psi_{\sigma_{1}}}_{B C_{1}} \otimes \ketbraq{1}_{C_{2}} \Big)\notag\\
&\hspace{7pt}+ \frac{1}{2} \ketbraq{w_-}_{W} \otimes \Big( \ketbraq{\psi_{\sigma_{0}}}_{B C_{1}} \otimes \ketbraq{0}_{C_{2}} - e^{i \theta} \ketbra{\psi_{\sigma_{0}}}{\psi_{\sigma_{1}}}_{BC_{1}} \otimes \ketbra{0}{1}_{C_{2}}\notag\\
&\hspace{7pt}- e^{- i \theta} \ketbra{\psi_{\sigma_{1}}}{\psi_{\sigma_{0}}}_{BC_{1}} \otimes \ketbra{1}{0}_{C_{2}} + \ketbraq{\psi_{\sigma_{1}}}_{B C_{1}} \otimes \ketbraq{1}_{C_{2}} \Big).
\end{align}
It is easy to verify that $\Tr_{B} \ketbra{\psi_{\sg_{u}}}{\psi_{\sg_{v}}}_{BC} = (\sqrt{\sg_{u}} \sqrt{\sg_{v}})_{C}$ for $u, v \in \{0, 1\}$. Therefore, we get
\begin{align}
\rho_{W C_{1} C_{2}} &= \frac{1}{2} \ketbraq{w_+} \otimes \Big( \sigma_{0} \otimes \ketbraq{0} + e^{i \theta} \sqrt{\sigma_{0}} \sqrt{\sigma_{1}} \otimes \ketbra{0}{1} + e^{- i \theta} \sqrt{\sigma_{1}} \sqrt{\sigma_{0}} \otimes \ketbra{1}{0} + \sigma_{1} \otimes \ketbraq{1} \Big)\notag\\
&\hspace{7pt}+ \frac{1}{2} \ketbraq{w_-} \otimes \Big( \sigma_{0} \otimes \ketbraq{0} - e^{i \theta} \sqrt{\sigma_{0}} \sqrt{\sigma_{1}} \otimes \ketbra{0}{1} - e^{- i \theta} \sqrt{\sigma_{1}} \sqrt{\sigma_{0}} \otimes \ketbra{1}{0} + \sigma_{1} \otimes \ketbraq{1} \Big).
\end{align}
Using the relation between the guessing probability and the trace distance from \eqref{eqnpguesstd1}, it is straightforward to find:
\begin{align}
2p_{\g}(W | C)_{\ket{\psi}} -1 &=\norm[\Big]{ e^{i \theta} \sqrt{\sigma_{0}} \sqrt{\sigma_{1}} \otimes \ketbra{0}{1} + e^{- i \theta} \sqrt{\sigma_{1}} \sqrt{\sigma_{0}} \otimes \ketbra{1}{0} }_{1} \notag\\
\label{eqn49596030}&= 2 \norm{\sqrt{\sigma_{0}} \sqrt{\sigma_{1}}}_{1}
\end{align}
It is worth noting that the phase $\theta$ does not appear in \eqref{eqn49596030}, i.e., all measurements in the $XY$ plane lead to the same guessing probability.
\end{proof}

\section{Testing coherence in a quantum beam splitter}\label{sct2A}

\subsection{Quantities sensitive to coherence}

Here we further elaborate on the treatment of the quantum beam splitter ($\qbs$) depicted in Fig.~5 of the main text. The $\qbs$ is modelled as a controlled-unitary, whose form is
\begin{equation}
\label{eqnUps1}
U_{PQ} = \dya{H}_P\ot \id_S+\dya{V}_P\ot U(R),\quad\text{with }
U(R) =
\begin{pmatrix}
 \sqrt{R}     & \sqrt{1-R}   \\
  \sqrt{1-R}    & -\sqrt{R} 
\end{pmatrix}.
\end{equation}
Here the photon's polarisation $P$ in the $\{ \ket{H},\ket{V}\}$ basis determines whether one applies the transformation associated with a (classical) beam splitter with reflection coefficient $R$. If $P$ is fed in as a superposition of the $\{ \ket{H},\ket{V}\}$ basis states then the beam splitter is said to be ``quantum''.

Our main goal in what follows is to show that our novel WPDR stated in the main text:
\begin{equation}
\label{eqnWPDRgood}
(\DC_i^P)^2 + \VC^2\leq 1
\end{equation}
captures the coherence of the beam splitter, whereas a weaker WPDR:
\begin{equation}
\label{eqnWPDRbad}
\DC_i^2 + \VC^2\leq 1
\end{equation}
does not. Here the different distinguishabilities are
\begin{align}
\label{eqndi1}\DC_i & = 2 p_{\g}(Z_i )_{D_0}-1, \\
\label{eqndpi1}\DC_i^P &= 2 p_{\g}(Z_i | P)_{D_0}-1,
\end{align}
where $P$ in \eqref{eqndpi1} refers to the final polarization after the $\qbs$. For simplicity, we will neglect any interaction with an external environment $\En$ in what follows, and hence conditioning these distinguishabilities on $\En$ is not necessary. (Since we are interested in demonstrating that \eqref{eqnWPDRgood} can capture coherence, it suffices to demonstrate it for a special case where $\En$ plays no role.) 

For comparison, we will also define ``decohered'' versions of $\DC_i$ and $\DC_i^P$, $\DC_i^{dec}$ and $\DC_i^{P,dec}$ respectively, where the latter correspond to feeding in a decohered version of the polarization state $\rho^{(2)}_P$, i.e., feeding in the corresponding classical mixture of $\ket{H}$ and $\ket{V}$ rather than a coherent superposition. Precisely this means replacing $\rho^{(2)}_P$ with $\ZC(\rho^{(2)}_P)$ where $\ZC(\cdot) = \dya{H}(\cdot)\dya{H} +\dya{V}(\cdot)\dya{V}$ is the quantum channel that decoheres the polarization state.

The first noteworthy point is that $\DC_i = \DC_i^{dec}$, hence measuring $\DC_i$ does not reveal the coherence in the $\qbs$. This is because $\DC_i$ is not conditioned on $P$, so we evaluate it on the reduced state obtained from tracing over $P$. But tracing over $P$ removes any dependence on the off-diagonal elements of $\rho^{(2)}_P$ in $\{\ket{H}, \ket{V}\}$ basis, since the unitary $U_{PQ}$ is controlled by the $\{\ket{H}, \ket{V}\}$ basis.

On the other hand we show that, in general, $\DC_i^P \neq \DC_i^{P,dec}$, so $\DC_i^P$ has the potential to reveal coherence. We also remark that the following hierarchy holds in general:
\begin{equation}
\DC_i \leq \DC_i^{P,dec} \leq \DC_i^P . 
\end{equation}
The first inequality holds because $\DC_i = \DC_i^{dec}$ and conditioning on $P$ can never decrease the guessing probability. The second inequality holds because the decoherence operation $\ZC(\cdot)$ commutes with $U_{PQ}$ and hence can be viewed as restricting the class of measurements over which one optimizes to evaluate the guessing probability.

For simplicity, let us consider a one-parameter family of input states $\rho^{(2)}_P = \dya{\al}$ with $\ket{\al} :=\cos \al \ket{H} + \sin \al \ket{V}$, which also happens to be the family considered in Ref.~\cite{Kaiser02112012} (see the next subsection). Thus, we have an open interferometer when $\al = 0$ and a closed one when $\al = 90 \deg$. For such states the visibility becomes
\begin{equation}
\VC = 2\sqrt{R(1-R)}\sin^2 (\al).
\end{equation}
Likewise, solving for the different distinguishabilities (see below for derivation) gives:
\begin{align}
\label{eqnDists1}\DC_i & = \abs{1 - 2 (\sin^2 \al) (1-R)}, \\
\label{eqnDists2}\DC_i^{P,dec} & = 1- (\sin^2 \al) \cdot (1-\abs{2 R - 1}),  \\
\label{eqnDists3}\DC_i^P & =  \sqrt{1 - 4R(1-R)(\sin^4 \al) }. 
\end{align}
Clearly these formulas indicate that, in general, $\DC_i^P \neq \DC_i^{P,dec}$, hence showing that $\DC_i^P$ reveals coherence. We can see a clear distinction between these distinguishabilities in Fig.~\ref{fgrsup1}, which considers the case of $R=0.4$.

\begin{figure}[t]
\begin{center}
\includegraphics[scale=1.2]{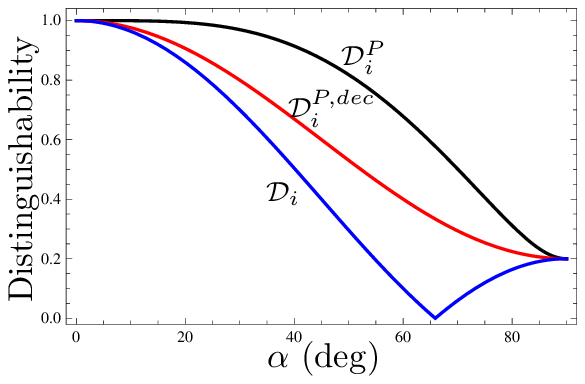}
\caption{%
Plot of $\DC_i$, $\DC_i^{P,dec}$, and $\DC_i^P$ as a function of $\al$, for $R=0.4$.\label{fgrsup1}}
\end{center}
\end{figure}

\subsection{Discussion of Ref.~\cite{Kaiser02112012}}

We remark that $\DC_i^P$ could be measured using polarization-resolving detection at the interferometer output, as in Fig.~5 of the main text, assuming one chooses the optimal polarization basis to measure. (See below where we explicitly solve for the optimal polarization basis to measure.) We note that the setup in Ref.~\cite{Kaiser02112012} has polarization-resolving detectors, and hence could measure $\DC_i^P$. However, the procedure outlined in \cite{Kaiser02112012} for measuring distinguishability corresponds to measuring our $\DC_i$. Figure~\ref{fgrsup2} shows our theoretical predictions for the situation in \cite{Kaiser02112012}, corresponding to $R=0.5$. At first sight our predictions appear to disagree with \cite{Kaiser02112012} in the sense that our Fig.~\ref{fgrsup2}B, which plots $\VC^2$, $\DC_i^2$, and $\VC^2 +\DC_i^2$, looks very different from the corresponding plot of these quantities in Fig.~4 of \cite{Kaiser02112012}. An explanation for the disagreement is that \cite{Kaiser02112012} may have actually plotted $\VC$, $\DC_i$, and $\VC +\DC_i$ in their Fig.~4. Indeed, their Fig.~4 looks similar to our predictions for $\VC$ and $\DC_i$ in Fig.~\ref{fgrsup2}A, and we have $\VC +\DC_i =1$ which is consistent with their Fig.~4. The authors of \cite{Kaiser02112012} have confirmed that their Fig.~4 plotted visibility and distinguishability as opposed to their squares \cite{TanzilliPrivComm}. We emphasize that this minor issue with their plot does not affect the conclusions of Ref.~\cite{Kaiser02112012}.

Since we predict that $\VC^2 +\DC_i^2$ can be strictly less than 1, then testing our novel relation $\VC^2 +(\DC_i^P)^2\leq 1$ can give a more stringent test of wave-particle duality. Indeed we show that for the setup in \cite{Kaiser02112012}, this relation is as strong as possible, i.e., it is satisfied with equality $\VC^2 +(\DC_i^P)^2 = 1$. This is depicted in Fig.~\ref{fgrsup2}C.

We remark that Ref.~\cite{Kaiser02112012} verified the coherence of their $\qbs$ by violating a Bell inequality, which involves verifying correlations in multiple bases. Our work here shows that one can verify the coherence of a $\qbs$ without switching bases in the experiment. Simply measuring the quantities appearing in our WPDR \eqref{eqnWPDRgood}, in particular $\DC_i^P$, can verify the coherence of the $\qbs$.

\begin{figure}[t]
\begin{center}
\includegraphics[scale=1.2]{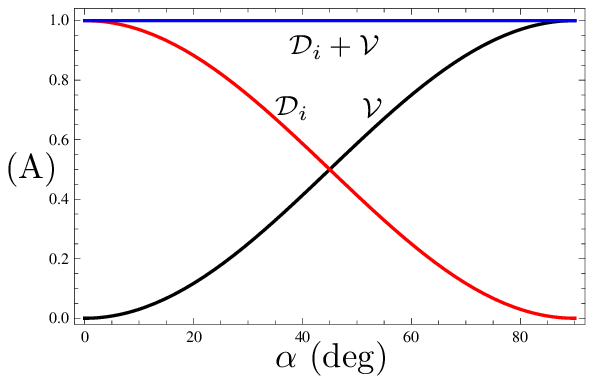}\\

\vspace{5pt}
\includegraphics[scale=1.2]{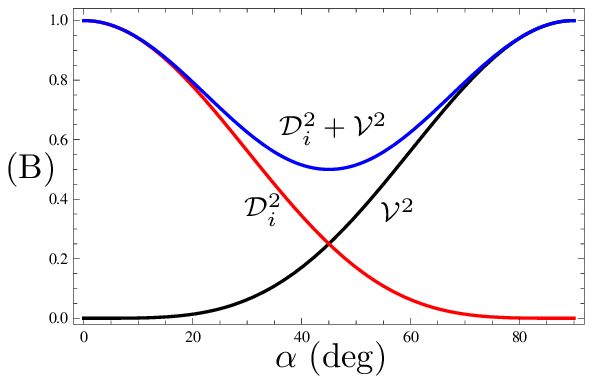}\\

\vspace{5pt}
\includegraphics[scale=1.2]{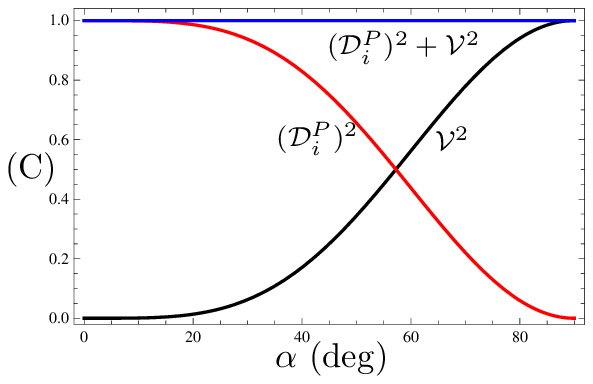}
\caption{%
For $R=0.5$, we plot (A) $\VC$, $\DC_i$, and $\VC +\DC_i$; (B) $\VC^2$, $\DC_i^2$, and $\VC^2 +\DC_i^2$; and (C) $\VC^2$, $(\DC_i^P)^2$, and $\VC^2 +(\DC_i^P)^2$. Notice that  $\VC^2 +(\DC_i^P)^2 = 1$, showing that our novel WPDR is perfectly tight.\label{fgrsup2}}
\end{center}
\end{figure}

\subsection{Derivation of distinguishability formulas}

The derivation of Eqs.~\eqref{eqnDists1}-\eqref{eqnDists3} proceeds as follows. As discussed in the Methods section of the main text, our retrodictive framework starts with the quanton $Q$ being maximally entangled to a register $Q'$, in the state $\ket{\Phi}_{Q Q'} = \frac{1}{\sqrt{2}} (\ket{00} + \ket{11})$. Meanwhile the polarization $P$ enters the $\qbs$ in the state $\ket{\al} : = c \ket{H} + s \ket{V}$, where $c = \cos \alpha$ and $s = \sin \alpha$. The action of the $\qbs$, given by $U_{PQ}$ in \eqref{eqnUps1}, results in the state:
\begin{equation}
\big( U_{PQ} \otimes \id_{Q'} \big) \ket{\al}_P \otimes \ket{\Phi}_{Q Q'} = c \ket{H}_{P} \otimes \frac{\ket{0}_{Q} \ket{0}_{Q'} + \ket{1}_{Q} \ket{1}_{Q'}}{\sqrt{2}} + s \ket{V}_{P} \otimes \frac{U(R) \ket{0}_{Q} \ket{0}_{Q'} + U(R) \ket{1}_{Q} \ket{1}_{Q'}}{\sqrt{2}}.
\end{equation}
Conditioning on detector $D_{0}$ clicking corresponds to measuring $Q$ in the standard basis and post-selecting on the outcome $0$. In addition, measuring $Q'$ to obtain $Z$ gives
\begin{gather}
\label{eqnrhozcp}
\bar{\rho}_{Z P}^{D_{0}} = \ketbraq{0}_{Z} \otimes \sigma_{P}^{0} + \ketbraq{1}_{Z} \otimes \sigma_{P}^{1},\\
\label{eqnsg0Psg1P}
\nbox{where} \sigma_{P}^{0} = \left(
\begin{array}{c c}
	c^{2} & s c \sqrt{R}\\
	s c \sqrt{R} & s^{2} R\\
\end{array} \right)
\nbox{and} \sigma_{P}^{1} = \left(
\begin{array}{c c}
	0 & 0\\
	0 & s^{2} (1 - R)\\
\end{array} \right).
\end{gather}
Now we compute the distinguishability by using the fact that it is the trace distance between the conditional states. For the polarization-enhanced distinguishability, this gives:
\begin{equation}
\DC_i^P = \norm{\sigma_{P}^{0} - \sigma_{P}^{1}}_{1} = \sqrt{c^{4} + s^{4} (2R - 1)^{2} + 2 (sc)^{2}}.
\end{equation}
Decohering $\sigma_{P}^{0}$ and $\sigma_{P}^{1}$ before calculating the trace norm leads to the decohered distinguishability:
\begin{equation}
\DC_i^{P,dec} = c^{2} + s^{2} \cdot \abs{2 R - 1}.
\end{equation}
Finally, to calculate the non-enhanced distinguishability we need to trace out the polarization register to get
\ca
\begin{equation}
\sigma_{C}^{0} = \left(
\begin{array}{c c}
	c^{2} + s^{2} R &\\
	& s^{2} (1 - R)\\
\end{array} \right) \nbox[6]{and}
\sigma_{C}^{1} = \left(
\begin{array}{c c}
	s^{2} (1 - R) &\\
	& c^{2} + s^{2} R\\
\end{array} \right)
\end{equation}
and then
\cb
\begin{equation}
\DC_i = \abs{\Tr \sigma_{P}^{0} - \Tr \sigma_{P}^{1}} = \abs{c^{2} + s^{2} (2R - 1)}.
\end{equation}

\subsection{Measuring $\DC_i^P$}

Measuring the polarization-enhanced distinguishability $\DC_i^P$ requires polarization-resolving detection at the output of the interferometer, as in Fig.~5 of the main text. As defined, $\DC_i^P$ corresponds to measuring the optimal polarization basis at the output, i.e., optimally helpful for guessing which path the photon took. We now solve for this optimal polarization basis. We remark that varying the polarization measurement basis could be accomplished by varying the angle of a half-wave plate inserted just prior to the PBS's in Fig.~5 of the main text.

Finding the optimal measurement is simply a two-state discrimination problem on a qubit, and the solution is well-known \cite{PhysRevA.70.022302}. Consider the (unnormalized) conditional states on system $P$ associated with $Z=0$ and $Z=1$, respectively denoted $\sg_P^0$ and $\sg_P^1$ and given by Eq.~\eqref{eqnsg0Psg1P}. (Note that both states are conditioned on detector $D_0$ clicking.) From \cite{PhysRevA.70.022302}, the optimal polarization basis to measure is given by the eigenvectors of the following Hermitian operator:
\begin{align}
O_P = \sg_P^{0} - \sg_P^{1}.
\end{align}
From \eqref{eqnsg0Psg1P}, we compute that this corresponds to the following polarization observable (represented as a matrix in the $\{\ket{H},\ket{V}\}$ basis):
\begin{align}
\widehat{O}_P =  \frac{1}{\DC_i^P}\left(
\begin{array}{c c}
	1-2(\sin^2 \al) R  & (\sin 2\al)\sqrt{R}\\
	(\sin 2\al)\sqrt{R} & - [1-2(\sin^2 \al)R ] \\
\end{array} \right),
\end{align}
where we have normalized the observable (hence the hat symbol) such that it squares to the identity. For example, choosing $R=0.5$ and $\al = 45 \deg$ (corresponding to an equal superposition of $\bs_2$ being ``absent'' and ``present'') gives
\begin{align}
\widehat{O}_P =  \frac{1}{\sqrt{3}}\left(
\begin{array}{c c}
	1  & \sqrt{2}\\
	\sqrt{2} & - 1 \\
\end{array} \right).
\end{align}


\begin{thebibliography}{44}
\expandafter\ifx\csname natexlab\endcsname\relax\def\natexlab#1{#1}\fi
\expandafter\ifx\csname bibnamefont\endcsname\relax
  \def\bibnamefont#1{#1}\fi
\expandafter\ifx\csname bibfnamefont\endcsname\relax
  \def\bibfnamefont#1{#1}\fi
\expandafter\ifx\csname citenamefont\endcsname\relax
  \def\citenamefont#1{#1}\fi
\expandafter\ifx\csname url\endcsname\relax
  \def\url#1{\texttt{#1}}\fi
\expandafter\ifx\csname urlprefix\endcsname\relax\def\urlprefix{URL }\fi
\providecommand{\bibinfo}[2]{#2}
\providecommand{\eprint}[2][]{\url{#2}}

\bibitem[{\citenamefont{Feynman}(1970)}]{Feynman70}
\bibinfo{author}{\bibfnamefont{R.~P.} \bibnamefont{Feynman}},
  \emph{\bibinfo{title}{Feynman Lectures on Physics}}
  (\bibinfo{publisher}{Addison Wesley, Longman}, \bibinfo{year}{1970}).

\bibitem[{\citenamefont{Englert}(1996)}]{EnglertPRL1996}
\bibinfo{author}{\bibfnamefont{B.-G.} \bibnamefont{Englert}},
  \bibinfo{journal}{Phys. Rev. Lett.} \textbf{\bibinfo{volume}{77}},
  \bibinfo{pages}{2154} (\bibinfo{year}{1996}).

\bibitem[{\citenamefont{Jaeger et~al.}(1995)\citenamefont{Jaeger, Shimony, and
  Vaidman}}]{JaegerEtAlPRA1995}
\bibinfo{author}{\bibfnamefont{G.}~\bibnamefont{Jaeger}},
  \bibinfo{author}{\bibfnamefont{A.}~\bibnamefont{Shimony}}, \bibnamefont{and}
  \bibinfo{author}{\bibfnamefont{L.}~\bibnamefont{Vaidman}},
  \bibinfo{journal}{Phys. Rev. A} \textbf{\bibinfo{volume}{51}},
  \bibinfo{pages}{54} (\bibinfo{year}{1995}).

\bibitem[{\citenamefont{Wootters and Zurek}(1979)}]{PhysRevD.19.473}
\bibinfo{author}{\bibfnamefont{W.~K.} \bibnamefont{Wootters}} \bibnamefont{and}
  \bibinfo{author}{\bibfnamefont{W.~H.} \bibnamefont{Zurek}},
  \bibinfo{journal}{Phys. Rev. D} \textbf{\bibinfo{volume}{19}},
  \bibinfo{pages}{473} (\bibinfo{year}{1979}).

\bibitem[{\citenamefont{Greenberger and Yasin}(1988)}]{Greenberger1988391}
\bibinfo{author}{\bibfnamefont{D.~M.} \bibnamefont{Greenberger}}
  \bibnamefont{and} \bibinfo{author}{\bibfnamefont{A.}~\bibnamefont{Yasin}},
  \bibinfo{journal}{Physics Letters A} \textbf{\bibinfo{volume}{128}},
  \bibinfo{pages}{391 } (\bibinfo{year}{1988}), ISSN \bibinfo{issn}{0375-9601}.

\bibitem[{\citenamefont{Englert et~al.}(2008)\citenamefont{Englert,
  Kaszlikowski, Kwek, and Chee}}]{EnglertIJQI2008}
\bibinfo{author}{\bibfnamefont{B.-G.} \bibnamefont{Englert}},
  \bibinfo{author}{\bibfnamefont{D.}~\bibnamefont{Kaszlikowski}},
  \bibinfo{author}{\bibfnamefont{L.~C.} \bibnamefont{Kwek}}, \bibnamefont{and}
  \bibinfo{author}{\bibfnamefont{W.~H.} \bibnamefont{Chee}},
  \bibinfo{journal}{International Journal of Quantum Information}
  \textbf{\bibinfo{volume}{06}}, \bibinfo{pages}{129} (\bibinfo{year}{2008}).

\bibitem[{\citenamefont{Liu et~al.}(2009)\citenamefont{Liu, Li, Yu, and
  Chen}}]{PhysRevA.79.052108}
\bibinfo{author}{\bibfnamefont{N.-L.} \bibnamefont{Liu}},
  \bibinfo{author}{\bibfnamefont{L.}~\bibnamefont{Li}},
  \bibinfo{author}{\bibfnamefont{S.}~\bibnamefont{Yu}}, \bibnamefont{and}
  \bibinfo{author}{\bibfnamefont{Z.-B.} \bibnamefont{Chen}},
  \bibinfo{journal}{Phys. Rev. A} \textbf{\bibinfo{volume}{79}},
  \bibinfo{pages}{052108} (\bibinfo{year}{2009}).

\bibitem[{\citenamefont{Huang et~al.}(2013)\citenamefont{Huang, W\"olk, Zhu,
  and Zubairy}}]{PhysRevA.87.022107}
\bibinfo{author}{\bibfnamefont{J.-H.} \bibnamefont{Huang}},
  \bibinfo{author}{\bibfnamefont{S.}~\bibnamefont{W\"olk}},
  \bibinfo{author}{\bibfnamefont{S.-Y.} \bibnamefont{Zhu}}, \bibnamefont{and}
  \bibinfo{author}{\bibfnamefont{M.~S.} \bibnamefont{Zubairy}},
  \bibinfo{journal}{Phys. Rev. A} \textbf{\bibinfo{volume}{87}},
  \bibinfo{pages}{022107} (\bibinfo{year}{2013}).

\bibitem[{\citenamefont{Qureshi}(2013)}]{Qureshi01042013}
\bibinfo{author}{\bibfnamefont{T.}~\bibnamefont{Qureshi}},
  \bibinfo{journal}{Progress of Theoretical and Experimental Physics}
  \textbf{\bibinfo{volume}{2013}} (\bibinfo{year}{2013}).

\bibitem[{\citenamefont{Li et~al.}(2012)\citenamefont{Li, Liu, and
  Yu}}]{PhysRevA.85.054101}
\bibinfo{author}{\bibfnamefont{L.}~\bibnamefont{Li}},
  \bibinfo{author}{\bibfnamefont{N.-L.} \bibnamefont{Liu}}, \bibnamefont{and}
  \bibinfo{author}{\bibfnamefont{S.}~\bibnamefont{Yu}}, \bibinfo{journal}{Phys.
  Rev. A} \textbf{\bibinfo{volume}{85}}, \bibinfo{pages}{054101}
  (\bibinfo{year}{2012}).

\bibitem[{\citenamefont{Englert and Bergou}(2000)}]{Englert2000337}
\bibinfo{author}{\bibfnamefont{B.-G.} \bibnamefont{Englert}} \bibnamefont{and}
  \bibinfo{author}{\bibfnamefont{J.~A.} \bibnamefont{Bergou}},
  \bibinfo{journal}{Optics Communications} \textbf{\bibinfo{volume}{179}},
  \bibinfo{pages}{337 } (\bibinfo{year}{2000}), ISSN \bibinfo{issn}{0030-4018}.

\bibitem[{\citenamefont{Banaszek et~al.}(2013)\citenamefont{Banaszek,
  Horodecki, Karpi{\'n}ski, and Radzewicz}}]{Banaszek:2013fk}
\bibinfo{author}{\bibfnamefont{K.}~\bibnamefont{Banaszek}},
  \bibinfo{author}{\bibfnamefont{P.}~\bibnamefont{Horodecki}},
  \bibinfo{author}{\bibfnamefont{M.}~\bibnamefont{Karpi{\'n}ski}},
  \bibnamefont{and}
  \bibinfo{author}{\bibfnamefont{C.}~\bibnamefont{Radzewicz}},
  \bibinfo{journal}{Nat Commun} \textbf{\bibinfo{volume}{4}}
  (\bibinfo{year}{2013}).

\bibitem[{\citenamefont{Jia et~al.}(2014)\citenamefont{Jia, Huang, Feng, Zhang,
  and Zhu}}]{JiaEtAlCPB2014}
\bibinfo{author}{\bibfnamefont{A.-A.} \bibnamefont{Jia}},
  \bibinfo{author}{\bibfnamefont{J.-H.} \bibnamefont{Huang}},
  \bibinfo{author}{\bibfnamefont{W.}~\bibnamefont{Feng}},
  \bibinfo{author}{\bibfnamefont{T.-C.} \bibnamefont{Zhang}}, \bibnamefont{and}
  \bibinfo{author}{\bibfnamefont{S.-Y.} \bibnamefont{Zhu}},
  \bibinfo{journal}{Chinese Physics B} \textbf{\bibinfo{volume}{23}},
  \bibinfo{pages}{30307} (\bibinfo{year}{2014}).

\bibitem[{\citenamefont{Englert et~al.}(1995)\citenamefont{Englert, Scully, and
  Walther}}]{Englert:1995ly}
\bibinfo{author}{\bibfnamefont{B.-G.} \bibnamefont{Englert}},
  \bibinfo{author}{\bibfnamefont{M.~O.} \bibnamefont{Scully}},
  \bibnamefont{and} \bibinfo{author}{\bibfnamefont{H.}~\bibnamefont{Walther}},
  \bibinfo{journal}{Nature} \textbf{\bibinfo{volume}{375}},
  \bibinfo{pages}{367} (\bibinfo{year}{1995}).

\bibitem[{\citenamefont{Storey et~al.}(1994)\citenamefont{Storey, Tan, Collett,
  and Walls}}]{Storey:1994zr}
\bibinfo{author}{\bibfnamefont{P.}~\bibnamefont{Storey}},
  \bibinfo{author}{\bibfnamefont{S.}~\bibnamefont{Tan}},
  \bibinfo{author}{\bibfnamefont{M.}~\bibnamefont{Collett}}, \bibnamefont{and}
  \bibinfo{author}{\bibfnamefont{D.}~\bibnamefont{Walls}},
  \bibinfo{journal}{Nature} \textbf{\bibinfo{volume}{367}},
  \bibinfo{pages}{626} (\bibinfo{year}{1994}).

\bibitem[{\citenamefont{Wiseman and Harrison}(1995)}]{Wiseman:1995ys}
\bibinfo{author}{\bibfnamefont{H.}~\bibnamefont{Wiseman}} \bibnamefont{and}
  \bibinfo{author}{\bibfnamefont{F.}~\bibnamefont{Harrison}},
  \bibinfo{journal}{Nature} \textbf{\bibinfo{volume}{377}},
  \bibinfo{pages}{584} (\bibinfo{year}{1995}).

\bibitem[{\citenamefont{Bohr}(1928)}]{Bohr1928}
\bibinfo{author}{\bibfnamefont{N.}~\bibnamefont{Bohr}},
  \bibinfo{journal}{Nature} \textbf{\bibinfo{volume}{121}},
  \bibinfo{pages}{580} (\bibinfo{year}{1928}).

\bibitem[{\citenamefont{Heisenberg}(1927)}]{Heisenberg}
\bibinfo{author}{\bibfnamefont{W.}~\bibnamefont{Heisenberg}},
  \bibinfo{journal}{Zeitschrift f\"ur Physik} \textbf{\bibinfo{volume}{43}},
  \bibinfo{pages}{172} (\bibinfo{year}{1927}).

\bibitem[{\citenamefont{Kennard}(1927)}]{kennard1927quantum}
\bibinfo{author}{\bibfnamefont{E.}~\bibnamefont{Kennard}}, \bibinfo{journal}{Z.
  Phys} \textbf{\bibinfo{volume}{44}}, \bibinfo{pages}{326}
  (\bibinfo{year}{1927}).

\bibitem[{\citenamefont{Robertson}(1929)}]{Robertson}
\bibinfo{author}{\bibfnamefont{H.~P.} \bibnamefont{Robertson}},
  \bibinfo{journal}{Phys. Rev.} \textbf{\bibinfo{volume}{34}},
  \bibinfo{pages}{163} (\bibinfo{year}{1929}).

\bibitem[{\citenamefont{Deutsch}(1983)}]{deutsch}
\bibinfo{author}{\bibfnamefont{D.}~\bibnamefont{Deutsch}},
  \bibinfo{journal}{Physical Review Letters} \textbf{\bibinfo{volume}{50}},
  \bibinfo{pages}{631} (\bibinfo{year}{1983}).

\bibitem[{\citenamefont{Bia{\l}ynicki-Birula and Mycielski}(1975)}]{BiaMyc75}
\bibinfo{author}{\bibfnamefont{I.}~\bibnamefont{Bia{\l}ynicki-Birula}}
  \bibnamefont{and}
  \bibinfo{author}{\bibfnamefont{J.}~\bibnamefont{Mycielski}},
  \bibinfo{journal}{Communications in Mathematical Physics}
  \textbf{\bibinfo{volume}{44}}, \bibinfo{pages}{129} (\bibinfo{year}{1975}).

\bibitem[{\citenamefont{Maassen and Uffink}(1988)}]{MaassenUffink}
\bibinfo{author}{\bibfnamefont{H.}~\bibnamefont{Maassen}} \bibnamefont{and}
  \bibinfo{author}{\bibfnamefont{J.~B.~M.} \bibnamefont{Uffink}},
  \bibinfo{journal}{Phys. Rev. Lett.} \textbf{\bibinfo{volume}{60}},
  \bibinfo{pages}{1103} (\bibinfo{year}{1988}).

\bibitem[{\citenamefont{{Wehner} and {Winter}}(2010)}]{EURreview1}
\bibinfo{author}{\bibfnamefont{S.}~\bibnamefont{{Wehner}}} \bibnamefont{and}
  \bibinfo{author}{\bibfnamefont{A.}~\bibnamefont{{Winter}}},
  \bibinfo{journal}{New J. Phys.} \textbf{\bibinfo{volume}{12}},
  \bibinfo{pages}{025009} (\bibinfo{year}{2010}).

\bibitem[{\citenamefont{Renes and Boileau}(2009)}]{RenesBoileau}
\bibinfo{author}{\bibfnamefont{J.~M.} \bibnamefont{Renes}} \bibnamefont{and}
  \bibinfo{author}{\bibfnamefont{J.-C.} \bibnamefont{Boileau}},
  \bibinfo{journal}{Phys. Rev. Lett.} \textbf{\bibinfo{volume}{103}},
  \bibinfo{pages}{020402} (\bibinfo{year}{2009}).

\bibitem[{\citenamefont{{Berta} et~al.}(2010)\citenamefont{{Berta},
  {Christandl}, {Colbeck}, {Renes}, and {Renner}}}]{BertaEtAl}
\bibinfo{author}{\bibfnamefont{M.}~\bibnamefont{{Berta}}},
  \bibinfo{author}{\bibfnamefont{M.}~\bibnamefont{{Christandl}}},
  \bibinfo{author}{\bibfnamefont{R.}~\bibnamefont{{Colbeck}}},
  \bibinfo{author}{\bibfnamefont{J.~M.} \bibnamefont{{Renes}}},
  \bibnamefont{and} \bibinfo{author}{\bibfnamefont{R.}~\bibnamefont{{Renner}}},
  \bibinfo{journal}{Nature Physics} \textbf{\bibinfo{volume}{6}},
  \bibinfo{pages}{659} (\bibinfo{year}{2010}).

\bibitem[{\citenamefont{Coles et~al.}(2011)\citenamefont{Coles, Yu, Gheorghiu,
  and Griffiths}}]{ColesEtAlPRA2011}
\bibinfo{author}{\bibfnamefont{P.~J.} \bibnamefont{Coles}},
  \bibinfo{author}{\bibfnamefont{L.}~\bibnamefont{Yu}},
  \bibinfo{author}{\bibfnamefont{V.}~\bibnamefont{Gheorghiu}},
  \bibnamefont{and} \bibinfo{author}{\bibfnamefont{R.~B.}
  \bibnamefont{Griffiths}}, \bibinfo{journal}{Phys. Rev. A}
  \textbf{\bibinfo{volume}{83}}, \bibinfo{pages}{062338}
  (\bibinfo{year}{2011}).

\bibitem[{\citenamefont{Coles et~al.}(2012)\citenamefont{Coles, Colbeck, Yu,
  and Zwolak}}]{ColesColbeckYuZwolak2012PRL}
\bibinfo{author}{\bibfnamefont{P.~J.} \bibnamefont{Coles}},
  \bibinfo{author}{\bibfnamefont{R.}~\bibnamefont{Colbeck}},
  \bibinfo{author}{\bibfnamefont{L.}~\bibnamefont{Yu}}, \bibnamefont{and}
  \bibinfo{author}{\bibfnamefont{M.}~\bibnamefont{Zwolak}},
  \bibinfo{journal}{Phys. Rev. Lett.} \textbf{\bibinfo{volume}{108}},
  \bibinfo{pages}{210405} (\bibinfo{year}{2012}).

\bibitem[{\citenamefont{{Tomamichel} and {Renner}}(2011)}]{TomRen2010}
\bibinfo{author}{\bibfnamefont{M.}~\bibnamefont{{Tomamichel}}}
  \bibnamefont{and} \bibinfo{author}{\bibfnamefont{R.}~\bibnamefont{{Renner}}},
  \bibinfo{journal}{Phys. Rev. Lett.} \textbf{\bibinfo{volume}{106}},
  \bibinfo{pages}{110506} (\bibinfo{year}{2011}).

\bibitem[{\citenamefont{Durr and Rempe}(2000)}]{DurrRempe2000}
\bibinfo{author}{\bibfnamefont{S.}~\bibnamefont{Durr}} \bibnamefont{and}
  \bibinfo{author}{\bibfnamefont{G.}~\bibnamefont{Rempe}},
  \bibinfo{journal}{American Journal of Physics} \textbf{\bibinfo{volume}{68}},
  \bibinfo{pages}{1021} (\bibinfo{year}{2000}).

\bibitem[{\citenamefont{Busch and Shilladay}(2006)}]{Busch20061}
\bibinfo{author}{\bibfnamefont{P.}~\bibnamefont{Busch}} \bibnamefont{and}
  \bibinfo{author}{\bibfnamefont{C.}~\bibnamefont{Shilladay}},
  \bibinfo{journal}{Physics Reports} \textbf{\bibinfo{volume}{435}},
  \bibinfo{pages}{1 } (\bibinfo{year}{2006}), ISSN \bibinfo{issn}{0370-1573}.

\bibitem[{\citenamefont{Konig et~al.}(2009)\citenamefont{Konig, Renner, and
  Schaffner}}]{KonRenSch09}
\bibinfo{author}{\bibfnamefont{R.}~\bibnamefont{Konig}},
  \bibinfo{author}{\bibfnamefont{R.}~\bibnamefont{Renner}}, \bibnamefont{and}
  \bibinfo{author}{\bibfnamefont{C.}~\bibnamefont{Schaffner}},
  \bibinfo{journal}{IEEE Trans. Inf. Theory} \textbf{\bibinfo{volume}{55}},
  \bibinfo{pages}{4337 } (\bibinfo{year}{2009}).

\bibitem[{\citenamefont{Ionicioiu and Terno}(2011)}]{PhysRevLett.107.230406}
\bibinfo{author}{\bibfnamefont{R.}~\bibnamefont{Ionicioiu}} \bibnamefont{and}
  \bibinfo{author}{\bibfnamefont{D.~R.} \bibnamefont{Terno}},
  \bibinfo{journal}{Phys. Rev. Lett.} \textbf{\bibinfo{volume}{107}},
  \bibinfo{pages}{230406} (\bibinfo{year}{2011}).

\bibitem[{\citenamefont{Kaiser et~al.}(2012)\citenamefont{Kaiser, Coudreau,
  Milman, Ostrowsky, and Tanzilli}}]{Kaiser02112012}
\bibinfo{author}{\bibfnamefont{F.}~\bibnamefont{Kaiser}},
  \bibinfo{author}{\bibfnamefont{T.}~\bibnamefont{Coudreau}},
  \bibinfo{author}{\bibfnamefont{P.}~\bibnamefont{Milman}},
  \bibinfo{author}{\bibfnamefont{D.~B.} \bibnamefont{Ostrowsky}},
  \bibnamefont{and} \bibinfo{author}{\bibfnamefont{S.}~\bibnamefont{Tanzilli}},
  \bibinfo{journal}{Science} \textbf{\bibinfo{volume}{338}},
  \bibinfo{pages}{637} (\bibinfo{year}{2012}).

\bibitem[{\citenamefont{Peruzzo et~al.}(2012)\citenamefont{Peruzzo, Shadbolt,
  Brunner, Popescu, and O'Brien}}]{Peruzzo02112012}
\bibinfo{author}{\bibfnamefont{A.}~\bibnamefont{Peruzzo}},
  \bibinfo{author}{\bibfnamefont{P.}~\bibnamefont{Shadbolt}},
  \bibinfo{author}{\bibfnamefont{N.}~\bibnamefont{Brunner}},
  \bibinfo{author}{\bibfnamefont{S.}~\bibnamefont{Popescu}}, \bibnamefont{and}
  \bibinfo{author}{\bibfnamefont{J.~L.} \bibnamefont{O'Brien}},
  \bibinfo{journal}{Science} \textbf{\bibinfo{volume}{338}},
  \bibinfo{pages}{634} (\bibinfo{year}{2012}).

\bibitem[{\citenamefont{Tang et~al.}(2013)\citenamefont{Tang, Li, Li, and
  Guo}}]{TangEtAlPhysRevA.88.014103}
\bibinfo{author}{\bibfnamefont{J.-S.} \bibnamefont{Tang}},
  \bibinfo{author}{\bibfnamefont{Y.-L.} \bibnamefont{Li}},
  \bibinfo{author}{\bibfnamefont{C.-F.} \bibnamefont{Li}}, \bibnamefont{and}
  \bibinfo{author}{\bibfnamefont{G.-C.} \bibnamefont{Guo}},
  \bibinfo{journal}{Phys. Rev. A} \textbf{\bibinfo{volume}{88}},
  \bibinfo{pages}{014103} (\bibinfo{year}{2013}).

\bibitem[{\citenamefont{Jacques et~al.}(2008)\citenamefont{Jacques, Wu,
  Grosshans, Treussart, Grangier, Aspect, and Roch}}]{PhysRevLett.100.220402}
\bibinfo{author}{\bibfnamefont{V.}~\bibnamefont{Jacques}},
  \bibinfo{author}{\bibfnamefont{E.}~\bibnamefont{Wu}},
  \bibinfo{author}{\bibfnamefont{F.}~\bibnamefont{Grosshans}},
  \bibinfo{author}{\bibfnamefont{F.}~\bibnamefont{Treussart}},
  \bibinfo{author}{\bibfnamefont{P.}~\bibnamefont{Grangier}},
  \bibinfo{author}{\bibfnamefont{A.}~\bibnamefont{Aspect}}, \bibnamefont{and}
  \bibinfo{author}{\bibfnamefont{J.-F.} \bibnamefont{Roch}},
  \bibinfo{journal}{Phys. Rev. Lett.} \textbf{\bibinfo{volume}{100}},
  \bibinfo{pages}{220402} (\bibinfo{year}{2008}).

\bibitem[{\citenamefont{Franson}(1989)}]{PhysRevLett.62.2205}
\bibinfo{author}{\bibfnamefont{J.~D.} \bibnamefont{Franson}},
  \bibinfo{journal}{Phys. Rev. Lett.} \textbf{\bibinfo{volume}{62}},
  \bibinfo{pages}{2205} (\bibinfo{year}{1989}).

\bibitem[{\citenamefont{Tomamichel et~al.}(2012)\citenamefont{Tomamichel, Lim,
  Gisin, and Renner}}]{TLGR}
\bibinfo{author}{\bibfnamefont{M.}~\bibnamefont{Tomamichel}},
  \bibinfo{author}{\bibfnamefont{C.~C.~W.} \bibnamefont{Lim}},
  \bibinfo{author}{\bibfnamefont{N.}~\bibnamefont{Gisin}}, \bibnamefont{and}
  \bibinfo{author}{\bibfnamefont{R.}~\bibnamefont{Renner}},
  \bibinfo{journal}{Nature Communications} \textbf{\bibinfo{volume}{3}},
  \bibinfo{pages}{634} (\bibinfo{year}{2012}).

\bibitem[{\citenamefont{Bj\"ork et~al.}(1999)\citenamefont{Bj\"ork,
  S\"oderholm, Trifonov, Tsegaye, and Karlsson}}]{PhysRevA.60.1874}
\bibinfo{author}{\bibfnamefont{G.}~\bibnamefont{Bj\"ork}},
  \bibinfo{author}{\bibfnamefont{J.}~\bibnamefont{S\"oderholm}},
  \bibinfo{author}{\bibfnamefont{A.}~\bibnamefont{Trifonov}},
  \bibinfo{author}{\bibfnamefont{T.}~\bibnamefont{Tsegaye}}, \bibnamefont{and}
  \bibinfo{author}{\bibfnamefont{A.}~\bibnamefont{Karlsson}},
  \bibinfo{journal}{Phys. Rev. A} \textbf{\bibinfo{volume}{60}},
  \bibinfo{pages}{1874} (\bibinfo{year}{1999}).

\bibitem[{\citenamefont{Bosyk et~al.}(2013)\citenamefont{Bosyk, Portesi, Holik,
  and Plastino}}]{BosykEtAl2013}
\bibinfo{author}{\bibfnamefont{G.~M.} \bibnamefont{Bosyk}},
  \bibinfo{author}{\bibfnamefont{M.}~\bibnamefont{Portesi}},
  \bibinfo{author}{\bibfnamefont{F.}~\bibnamefont{Holik}}, \bibnamefont{and}
  \bibinfo{author}{\bibfnamefont{A.}~\bibnamefont{Plastino}},
  \bibinfo{journal}{Physica Scripta} \textbf{\bibinfo{volume}{87}},
  \bibinfo{pages}{065002} (\bibinfo{year}{2013}).

\bibitem[{\citenamefont{Buscemi et~al.}(2014)\citenamefont{Buscemi, Hall,
  Ozawa, and Wilde}}]{PhysRevLett.112.050401}
\bibinfo{author}{\bibfnamefont{F.}~\bibnamefont{Buscemi}},
  \bibinfo{author}{\bibfnamefont{M.~J.~W.} \bibnamefont{Hall}},
  \bibinfo{author}{\bibfnamefont{M.}~\bibnamefont{Ozawa}}, \bibnamefont{and}
  \bibinfo{author}{\bibfnamefont{M.~M.} \bibnamefont{Wilde}},
  \bibinfo{journal}{Phys. Rev. Lett.} \textbf{\bibinfo{volume}{112}},
  \bibinfo{pages}{050401} (\bibinfo{year}{2014}).

\bibitem[{\citenamefont{Ekert et~al.}(1992)\citenamefont{Ekert, Rarity,
  Tapster, and Massimo~Palma}}]{PhysRevLett.69.1293}
\bibinfo{author}{\bibfnamefont{A.~K.} \bibnamefont{Ekert}},
  \bibinfo{author}{\bibfnamefont{J.~G.} \bibnamefont{Rarity}},
  \bibinfo{author}{\bibfnamefont{P.~R.} \bibnamefont{Tapster}},
  \bibnamefont{and}
  \bibinfo{author}{\bibfnamefont{G.}~\bibnamefont{Massimo~Palma}},
  \bibinfo{journal}{Phys. Rev. Lett.} \textbf{\bibinfo{volume}{69}},
  \bibinfo{pages}{1293} (\bibinfo{year}{1992}).

\bibitem[{\citenamefont{Nielsen and Chuang}(2000)}]{NieChu00}
\bibinfo{author}{\bibfnamefont{M.~A.} \bibnamefont{Nielsen}} \bibnamefont{and}
  \bibinfo{author}{\bibfnamefont{I.~L.} \bibnamefont{Chuang}},
  \emph{\bibinfo{title}{Quantum Computation and Quantum Information}}
  (\bibinfo{publisher}{Cambridge University Press},
  \bibinfo{address}{Cambridge}, \bibinfo{year}{2000}), \bibinfo{edition}{5th}
  ed.


\bibitem[{\citenamefont{Helstrom}(1976)}]{helstrom76}
\bibinfo{author}{\bibfnamefont{C.~W.} \bibnamefont{Helstrom}},
  \emph{\bibinfo{title}{{Quantum detection and estimation theory}}}
  (\bibinfo{publisher}{Academic Press}, \bibinfo{address}{New York, USA},
  \bibinfo{year}{1976}), ISBN \bibinfo{isbn}{0123400503}.


\bibitem[{\citenamefont{King and Ruskai}(2001)}]{King01}
\bibinfo{author}{\bibfnamefont{C.}~\bibnamefont{King}} \bibnamefont{and}
  \bibinfo{author}{\bibfnamefont{M.}~\bibnamefont{Ruskai}},
  \bibinfo{journal}{IEEE Trans. Inf. Theory} \textbf{\bibinfo{volume}{47}},
  \bibinfo{pages}{192} (\bibinfo{year}{2001}).

\bibitem[{\citenamefont{{Tanzilli}}()}]{TanzilliPrivComm}
\bibinfo{author}{\bibfnamefont{S.}~\bibnamefont{{Tanzilli}}},
  \bibinfo{note}{private communication.}

\bibitem[{\citenamefont{Herzog and Bergou}(2004)}]{PhysRevA.70.022302}
\bibinfo{author}{\bibfnamefont{U.}~\bibnamefont{Herzog}} \bibnamefont{and}
  \bibinfo{author}{\bibfnamefont{J.~A.} \bibnamefont{Bergou}},
  \bibinfo{journal}{Phys. Rev. A} \textbf{\bibinfo{volume}{70}},
  \bibinfo{pages}{022302} (\bibinfo{year}{2004}).



\end{thebibliography}
\end{document}